\documentclass{article}
\usepackage{style}
\usepackage{shortcuts}

\title{All-Pairs Shortest Paths with Few Weights per Node}
\author{
    Amir Abboud\thanks{Weizmann Institute of Science. This work is part of the project CONJEXITY that has received funding from the European Research Council (ERC) under the European Union's Horizon Europe research and innovation programme (grant agreement No.~101078482). Supported by an Alon scholarship and a research grant from the Center for New Scientists at the Weizmann Institute of Science. Part of this work was done while visiting INSAIT, Sofia University ``St. Kliment Ohridski''.}\and
    Nick Fischer\thanks{INSAIT, Sofia University ``St. Kliment Ohridski''. Partially funded by the Ministry of Education and Science of Bulgaria's support for INSAIT, Sofia University ``St. Kliment Ohridski'' as part of the Bulgarian National Roadmap for Research Infrastructure. Part of this work was done while at the author was at the Weizmann Institute of Science.}\and
    Ce Jin\thanks{Massachusetts Institute of Technology. Supported by the Jane Street Graduate Research Fellowship, NSF grant CCF-2330048, and a Simons Investigator Award.}\and
    Virginia Vassilevska Williams\thanks{Massachusetts Institute of Technology. Supported by NSF Grant CCF-2330048, BSF Grant 2020356 and a Simons Investigator Award.}\and
    Zoe Xi\thanks{Massachusetts Institute of Technology.}}
\date{}

\begin{document}

\maketitle
\begin{abstract}
\noindent
We study the central All-Pairs Shortest Paths (APSP) problem under the restriction that there are at most $d$ distinct weights on the outgoing edges from every node.
For $d=n$ this is the classical (unrestricted) APSP problem that is hypothesized to require cubic time $n^{3-o(1)}$, and at the other extreme, for $d=1$, it is equivalent to the \emph{Node-Weighted} APSP problem.
We present new algorithms that achieve the following results:

\begin{itemize}
\setlength\parskip{0pt plus 1pt}
\setlength\parindent{1.6em}
\item Node-Weighted APSP can be solved in time \smash{$\widetilde{O}(n^{(3+\omega)/2}) = \widetilde{O}(n^{2.686})$}, improving on the 15-year-old subcubic bounds \smash{$\widetilde{O}(n^{(9+\omega)/4}) = \widetilde{O}(n^{2.843})$} [Chan; STOC~'07] and \smash{$\widetilde{O}(n^{2.830})$} [Yuster; SODA~'09]. This positively resolves the question of whether Node-Weighted APSP is an ``intermediate'' problem in the sense of having complexity $n^{2.5+o(1)}$ if $\omega=2$, in which case it also matches an $n^{2.5-o(1)}$ conditional lower bound.

\item For up to $d \leq n^{3-\omega-\epsilon}$ distinct weights per node (where $\epsilon > 0$), the problem can be solved in subcubic time $O(n^{3-f(\epsilon)})$ (where $f(\epsilon) > 0$). In particular, assuming that $\omega = 2$, we can tolerate any sublinear number of distinct weights per node $d \leq n^{1-\epsilon}$, whereas previous work [Yuster;~SODA~'09] could only handle $d \leq n^{1/2-\epsilon}$ in subcubic time. This promotes our understanding of the APSP hypothesis showing that the hardest instances must exhaust a linear number of weights per node. With the current bounds on~$\omega$, we achieve a subcubic algorithm for $d \leq n^{0.628}$ whereas previously a subcubic running time could only be achieved for $d \leq n^{0.384}$.

Our result also applies to the All-Pairs Exact Triangle problem, thus generalizing a result of Chan and Lewenstein on ``Clustered 3SUM'' from arrays to matrices. Notably, our technique constitutes a rare application of additive combinatorics in graph algorithms.
\end{itemize}
We complement our positive results with simple hardness reductions even for \emph{undirected} graphs. Interestingly, under fine-grained assumptions, the complexity in the undirected case jumps from $\tilde O(n^{\omega})$ for $d=1$ to $n^{2.5-o(1)}$ for $d \geq 2$.
\end{abstract}

\thispagestyle{empty}
\clearpage
\setcounter{page}{1}

\section{Introduction} \label{sec:intro}
The classical All-Pairs Shortest-Paths problem (APSP) asks to compute the distance between all $\binom{n}{2}$ pairs of nodes in an $n$-node graph with integer edge weights. 
In this paper, we investigate its time complexity under the natural restriction that there are few distinct weights on the edges touching each node.
(For simplicity of exposition, we have chosen to present the problem on \emph{directed} graphs. We remark that the open questions and our results apply to the undirected setting as well. In Section~\ref{sec:undir}, we discuss some interesting distinctions between the two cases.)

\begin{definition}[$d$-Weights APSP]
Given a directed edge-weighted\footnote{Here and throughout we assume that all edge weights are integers in the range $\set{-M, \dots, M}$ where $M = n^c$ for some constant~$c$. In fact, all results in this paper also work when $M$ is a $\polylog(n)$-bit integer.} graph $G = (V, E)$ such that for all nodes $v \in V$ there are at most $d$ distinct weights on its outgoing edges (i.e., $\abs{\set{w(v, u) : (v, u) \in E}} \leq d$), compute the distances between all pairs of nodes $u, v \in V$.
\end{definition}

Note that when $d=n$ the restriction is trivial and we get the unrestricted APSP problem; our interest is in the regime $1 \leq d \ll n$.

The first motivation for studying this problem stems from APSP's role as a central problem in algorithmic research.
It is natural to expect the graphs that arise in APSP's countless applications to often use a limited number of weights. 
In fact, at the extreme where $d=1$, we get the fundamental \emph{Node-Weighted APSP} problem---if all edges leaving a node have the same weight, then it is effectively a node weight.

The second motivation stems from APSP's role in fine-grained complexity, where its hypothesized hardness serves as one of the foundational pillars:
The hypothesis that APSP requires $n^{3-o(1)}$ time\footnote{All algorithms and hardness hypotheses are assumed to be in the word-RAM model of computation with $O(\log n)$-bit words.} has been used to classify the complexity of dozens of problems on graphs, matrices, formal languages, machine learning, and beyond (e.g.\ \cite{WilliamsW18,Saha15,BackursT17,BringmannGMW20,AbboudGW23,ChanWX21,AbboudW14}).
To promote our understanding of the hypothesis and our confidence in its corollaries, it is important to investigate on which instances APSP is hard.

\paragraph{Previous Work and Main Questions.}
The first subcubic time algorithm for Node-Weighted APSP (the case~\makebox{$d=1$}) was given by Chan \cite{Chan10}; the precise running time is~\smash{$\widetilde{O}(n^{(9+\omega)/4})$} which is~\smash{$\widetilde{O}(n^{2.75})$} if $\omega=2$.
Here, $\omega<2.371339$ is the fast matrix multiplication exponent \cite{AlmanDVXXZ25}.
Shortly after, Yuster \cite{Yuster09} obtained an improved algorithm which with the current bounds on rectangular matrix multiplication runs in $O(n^{2.830})$ time, but still runs in~\smash{$\widetilde{O}(n^{2.75})$} time if $\omega=2$.

In the same paper, Yuster also defined the $d$-Weights APSP problem for general $d$ and gave an algorithm that, when $\omega=2$, solves the problem in $\widetilde{O}(\sqrt{d}n^{2.75})$ time. 
If $\omega>2$ the running time is more involved, as it uses a different rectangular matrix multiplication exponent for each $d$\footnote{The running time is $\widetilde{O}(n^{3-(\gamma-\rho)/2})$ where $d=n^\rho$ and for $\gamma$ defined as the solution of $\omega(1+\rho,1,1+\gamma)=3+\rho-\gamma$ and $\rho<\mu$ for $\omega(1+\mu,1,1+\mu)=3$. Here $\omega(a,b,c)$ is the exponent of $n^a\times n^b$ by $n^b\times n^c$ matrix multiplication.}.
Yuster's running time is subcubic as long as $d=n^{1/2-\eps}$ for some constant $\eps>0$ if $\omega=2$, and as long as $d\leq n^{0.341}$ for the current bounds on rectangular matrix multiplication.

Meanwhile, on the hardness side, we know that essentially cubic time is required under the APSP hypothesis when $d=n$. Moreover, for all $d \geq 1$ $d$-Weights APSP is at least as hard as the \emph{unweighted} directed APSP problem which is conjectured to require $n^{2.5-o(1)}$ time; this is known as the \emph{u-dir-APSP Hypothesis}~\cite{ChanWX21,ChanWX23}.

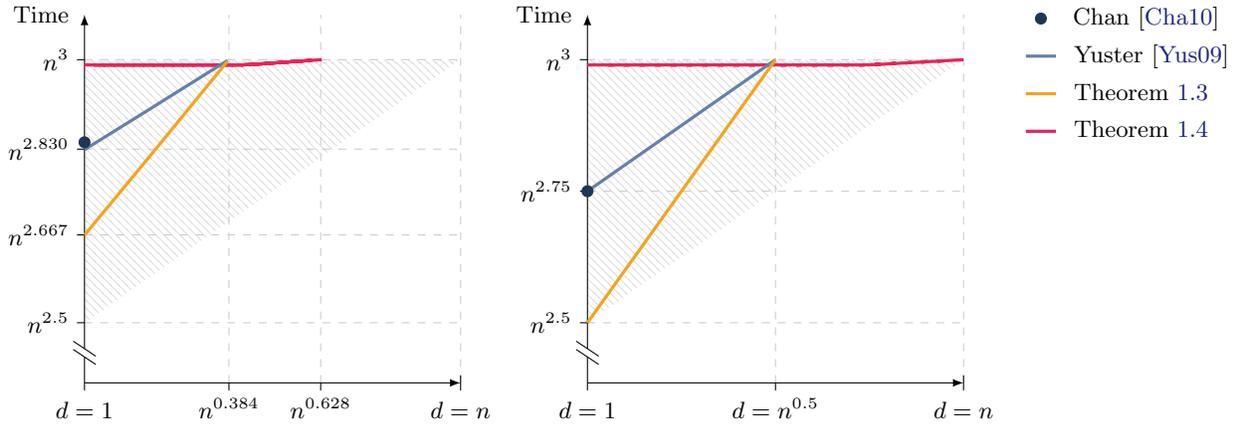
\begin{figure}[t]
\caption{Plots the time complexity of the $d$-Weights APSP problem. In the left diagram we have used the state-of-the-art bounds on fast matrix multiplication~\cite{AlmanDVXXZ25}; in the right diagram we assume $\omega = 2$. The hatched region illustrates the interesting region between the trivial $\Order(n^3)$ upper bound and the $d^{0.5} n^{2.5-\order(1)}$ conditional lower bound (\cref{thm:dapsp}). Also note that the red line depicting our \cref{thm:apsp-few-weights} is only conceptual in that it guarantees a subcubic running time $\Order(n^{3-\epsilon})$ for some small non-explicit constant~\smash{$\epsilon > 0$}.}
\label{fig:plots}
\rule{\textwidth}{0.5pt}

\bigskip
\begin{minipage}{\textwidth}
\begin{minipage}[t]{6cm}
\begin{tikzpicture}[%
    x=5cm, y=7cm,
    >=latex,
    every node/.append style={inner sep=0pt},
    fill between/on layer={background},
]
    \def\offsetlen{.8cm}
    \def\offsetsep{.11cm}
    \def\overshoot{.6cm}

    \draw (0, 2.5) ++(0, -\offsetlen) -- ++(0, -.1cm) node[below=.1cm] {\small $d=1\vphantom{n^0}$};
    \draw[dashed, black!20] (0.384, 2.5) ++(0, -\offsetlen) -- (0.384, 3) -- ++(0, \overshoot);
    \draw (0.384, 2.5) ++(0, -\offsetlen) -- ++(0, -.1cm) node[below=.1cm] {\small $n^{0.384}$};
    \draw[dashed, black!20] (0.628, 2.5) ++(0, -\offsetlen) -- (0.628, 3) -- ++(0, \overshoot);
    \draw (0.628, 2.5) ++(0, -\offsetlen) -- ++(0, -.1cm) node[below=.1cm] {\small $n^{0.628}$};
    \draw[dashed, black!20] (1, 2.5) ++(0, -\offsetlen) -- (1, 3) -- ++(0, \overshoot);
    \draw (1, 2.5) ++(0, -\offsetlen) ++(0, .1cm) -- ++(0, -.2cm) node[below=.1cm] {\small $d=n\vphantom{n^0}$};

    \draw[dashed, black!20] (0, 2.5) -- (1, 2.5);
    \draw (0, 2.5) -- ++(-.1cm, 0) node[left=.1cm] {\small $n^{2.5}$};
    \draw[dashed, black!20] (0, 2.667) -- (1, 2.667);
    \draw (0, 2.667) -- ++(-.1cm, 0) node[left=.1cm] {\small $n^{2.667}$};
    \draw[dashed, black!20] (0, 2.830) -- (1, 2.830);
    \draw (0, 2.830) -- ++(-.1cm, 0) node[left=.1cm] {\small $n^{2.830}$};
    \draw[dashed, black!20] (0, 3) -- (1, 3);
    \draw (0, 3) -- ++(-.1cm, 0) node[left=.1cm] {\small $n^3$};

    \draw[->] (0, 2.5) ++(0, -\offsetlen) -- ++(1, 0);
    \draw (0, 2.5) -- ++(0, -.5 * \offsetlen + 0.5 * \offsetsep);
    \draw (0, 2.5) ++(0, -\offsetlen) -- ++(0, .5 * \offsetlen - 0.5 * \offsetsep);
    \draw (0, 2.5) ++(-.15cm, -.5 * \offsetlen + .1cm + 0.5 * \offsetsep) -- ++(.3cm, -.2cm);
    \draw (0, 2.5) ++(-.15cm, -.5 * \offsetlen + .1cm - 0.5 * \offsetsep) -- ++(.3cm, -.2cm);
    \draw[->] (0, 2.5) -- (0, 3) -- ++(0, \overshoot) node[left=.2cm] {\small Time};

    \path[domain=0:1, variable=\d, very thick, name path=LB] plot ({\d}, {2.5+\d/2});
    \path[domain=0:1, variable=\d, very thick, name path=UB] plot ({\d}, {3});
    \tikzfillbetween[of=UB and LB]{gray, opacity=0.5, pattern=north west lines, pattern color=gray};

    \draw[domain=0:0.629, variable=\d, very thick, Red, rounded corners] plot ({\d}, {max(2.99,2.99+(\d-0.429)/0.2*0.01)});

    \draw[domain=0:0.38, variable=\d, very thick, PaleBlue] plot file {figures/plot_yuster.dat};

    \fill[DarkBlue] (0, 2.843) circle (0.08cm);

    \draw[domain=0:0.38, variable=\d, very thick, Orange] plot file {figures/plot_afjvx.dat};
\end{tikzpicture}
\end{minipage}
\hfill
\begin{minipage}[t]{6cm}
\begin{tikzpicture}[%
    x=5cm, y=7cm,
    >=latex,
    every node/.append style={inner sep=0pt},
    fill between/on layer={background},
]
    \def\offsetlen{.8cm}
    \def\offsetsep{.11cm}
    \def\overshoot{.6cm}

    \draw (0, 2.5) ++(0, -\offsetlen) -- ++(0, -.1cm) node[below=.1cm] {\small$d=1\vphantom{n^0}$};
    \draw[dashed, black!20] (0.5, 2.5) ++(0, -\offsetlen) -- (0.5, 3) -- ++(0, \overshoot);
    \draw (0.5, 2.5) ++(0, -\offsetlen) -- ++(0, -.1cm) node[below=.1cm] {\small $d=n^{0.5}$};
    \draw[dashed, black!20] (1, 2.5) ++(0, -\offsetlen) -- (1, 3) -- ++(0, \overshoot);
    \draw (1, 2.5) ++(0, -\offsetlen) ++(0, .1cm) -- ++(0, -.2cm) node[below=.1cm] {\small$d=n\vphantom{n^0}$};

    \draw[dashed, black!20] (0, 2.5) -- (1, 2.5);
    \draw (0, 2.5) -- ++(-.1cm, 0) node[left=.1cm] {\small $n^{2.5}$};
    \draw[dashed, black!20] (0, 2.75) -- (1, 2.75);
    \draw (0, 2.75) -- ++(-.1cm, 0) node[left=.1cm] {\small $n^{2.75}$};
    \draw[dashed, black!20] (0, 3) -- (1, 3);
    \draw (0, 3) -- ++(-.1cm, 0) node[left=.1cm] {\small $n^3$};

    \draw[->] (0, 2.5) ++(0, -\offsetlen) -- ++(1, 0);
    \draw (0, 2.5) -- ++(0, -.5 * \offsetlen + 0.5 * \offsetsep);
    \draw (0, 2.5) ++(0, -\offsetlen) -- ++(0, .5 * \offsetlen - 0.5 * \offsetsep);
    \draw (0, 2.5) ++(-.15cm, -.5 * \offsetlen + .1cm + 0.5 * \offsetsep) -- ++(.3cm, -.2cm);
    \draw (0, 2.5) ++(-.15cm, -.5 * \offsetlen + .1cm - 0.5 * \offsetsep) -- ++(.3cm, -.2cm);
    \draw[->] (0, 2.5) -- (0, 3) -- ++(0, \overshoot) node[left=.2cm] {\small Time};

    \path[domain=0:1, variable=\d, very thick, name path=LB] plot ({\d}, {2.5+\d/2});
    \path[domain=0:1, variable=\d, very thick, name path=UB] plot ({\d}, {3});
    \tikzfillbetween[of=UB and LB]{gray, opacity=0.5, pattern=north west lines, pattern color=gray};

    \draw[domain=0:1, variable=\d, very thick, Red, rounded corners] plot ({\d}, {max(2.99,2.99+(\d-0.75)/0.25*0.01)});

    \draw[domain=0:0.5, variable=\d, very thick, PaleBlue] plot ({\d}, {2.75+\d/2});

    \fill[DarkBlue] (0, 2.75) circle (0.08cm);

    \draw[domain=0:0.5, variable=\d, very thick, Orange] plot ({\d}, {2.5+\d});
\end{tikzpicture}
\end{minipage}
\hfill
\begin{minipage}[t]{3cm}
\vspace{-5.6cm}
\begin{tikzpicture}
    \path[fill=DarkBlue] (0, 0) node[anchor=west] {\small Chan~\cite{Chan10}} ++(-.3cm, 0) circle (0.08cm);
    \path[draw=PaleBlue, very thick] (0, -0.5cm) node[anchor=west] {\small Yuster~\cite{Yuster09}} ++(-.1cm, 0) -- ++(-.4cm, 0);
    \path[draw=Orange, very thick] (0, -1cm) node[anchor=west] {\small\cref{thm:apsp-few-weights-naive}\vphantom{[}} ++(-.1cm, 0) -- ++(-.4cm, 0);
    \path[draw=Red, very thick] (0, -1.5cm) node[anchor=west] {\small\cref{thm:apsp-few-weights}\vphantom{[}} ++(-.1cm, 0) -- ++(-.4cm, 0);
\end{tikzpicture}
\end{minipage}
\end{minipage}

\bigskip
\rule{\textwidth}{0.5pt}
\end{figure}

This state of affairs (plotted in blue in Figure \ref{fig:plots}) leaves a large polynomial gap for all sublinear $d$. While at the conceptual level it delivers the message that $d$-Weights APSP tends to become easier as $d$ decreases, it does so only loosely. To see this, let us highlight what are perhaps the two most important open questions. 

For this discussion, let us assume that $\omega=2$. Two reasons justify this: first, it lets us simplify and sharpen the story when comparing running time bounds, and second, this is a working assumption in much of fine-grained complexity in the sense that the central hypotheses and their consequences are expected to remain valid even if $\omega$ becomes $2$. When presenting our results in Section~\ref{sec:results}, we will specify the bounds under the current $\omega$ as well.

The first obvious open question is closing the gap for the $d=1$ case between the $\widetilde{O}(n^{2.75})$ upper bound and the $n^{2.5-o(1)}$ conditional lower bound. 
Once a cubic time bound is broken, the typical next question is whether the complexity can be reduced to $n^2$. 
However, it is common for progress to stop at $n^{2.5}$, and recent work \cite{LincolnPW20} has coined the term ``intermediate problems'' for this class (while also giving conditional lower bounds showing that $n^{2.5}$ is a barrier). 

Given that an $n^{2.5-o(1)}$ lower bound already exists for Node-Weighted APSP (under the u-dir-APSP hypothesis), the question becomes whether it is in this intermediate class or whether it is a ``subcubic but not intermediate'' problem.

\begin{question}
\label{oq1}
Is Node-Weighted APSP an ``intermediate problem''? I.e.\ can it be solved in $n^{2.5+o(1)}$ time if $\omega=2$?
\end{question}

The class of intermediate problems already contains a host of important problems on graphs (e.g.\ All-Pairs Bottleneck-Paths, All-Pairs Earliest-Arrivals)~\cite{DuanP09,VassilevskaWY09,DuanJW19}, matrices (e.g.\ Equality and Dominance Product, Min-Witness Product, Min-Plus Product in bounded-difference or monotone matrices)~\cite{Matousek91,KowalukL05,ChiDX022,Durr23}, and sequences (e.g.\ RNA Folding and Dyck Edit Distance)~\cite{BringmannGSW16,ChiDX022}.
In many of these cases, the $n^{2.5}$ bound was achieved a few years after the cubic bound was broken.
Thus, it is remarkable that the $n^{2.75}$ bound for Node-Weighted APSP has remained unbeaten for 15 years.
In fact, while the time hierarchy theorem gives the existence of (artificial) problems with complexity between $n^{2.5}$ and $n^{3}$, and while there are more natural graph problems with such complexity under fine-grained assumptions \cite{AbboudWY18}, we are aware of only very few other well-studied problems which have been shown to be in truly subcubic time but whose running time has not been made intermediate. (In fact, only one other example related to shortest paths comes to mind, \cite{ChechikZ24,WilliamsWX22}.).

Needless to say, understanding the node-weighted case has been a primary concern for many other basic graph problems such as triangle detection \cite{VassilevskaW06,CzumajL09} and minimum cut \cite{HenzingerRG00,HassinL07,AbboudKT21}. 

\medskip

The second open question is whether the hard instances for the APSP conjecture could have a sublinear number of distinct (edge) weights touching every node, i.e.\ whether the $d=n^{1-\eps}$ case is in truly subcubic time.
The running time of Yuster's algorithm becomes cubic already when $d=\sqrt{n}$.
A priori, it may very well be that this case is APSP-hard; e.g.\ perhaps we can take an arbitrary instance of APSP and ``sparsify'' or compress the weights so that only $\sqrt{n}$ weights are used in each neighborhood.

\begin{question}
\label{oq2}
Can APSP be solved in truly subcubic $O(n^{3-\delta})$ time if there is only a sublinear $d=n^{1-\eps}$ number of distinct edge weights per node (and $\omega=2$)?
\end{question}

This question is part of the general quest towards either refuting the APSP hypothesis or gaining a deeper understanding of the instances that make the problem hard.
This quest is particularly important for APSP because, unlike the other two central hypotheses in the field (3SUM and OV), where a natural, simple distribution appears to give hard instances, this is not the case for APSP.
Perhaps the first question that comes to mind is whether the weights in the hard instance could be small, i.e.\ if all edge-weights are in a set~\makebox{$\{-M,\ldots,M\}$} where $M$ is small. Thanks to now classical algorithms that, if $\omega=2$, run in time~\makebox{$M \cdot n^{2+o(1)}$} for undirected graphs~\cite{ShoshanZ99,AlonGM97} and time $\sqrt{M} \cdot n^{2.5+o(1)}$ for directed graphs~\cite{Zwick02}, we know that $M$ cannot be sublinear in the hard cases. (In fact, the ``Strong APSP Hypothesis'' \cite{ChanWX23} speculates that there are hard instances with $M=\Theta(n)$.)
Next, one may ask if, in the hard instances, the weights could come from a smaller set while still being large in value, and this is exactly Open Question~\ref{oq2}.

\subsection{Our Results}
\label{sec:results}

In this paper, we advance our understanding of the complexity of the $d$-Weights APSP problem by positively resolving Open Questions~\ref{oq1} and~\ref{oq2}.

Our first result is a faster algorithm for Node-Weighted APSP, proving its containment in the class of ``intermediate'' problems.

\begin{restatable}[Node-Weighted APSP]{theorem}{thmapspnodeweighted} \label{thm:apsp-node-weighted}
There is a deterministic algorithm for Node-Weighted APSP running in time~\smash{$\widetilde\Order(n^{(3+\omega)/2}) = \widetilde\Order(n^{2.686})$}. In fact, using rectangular fast matrix multiplication this can be improved to~\smash{$\widetilde\Order(n^{2.667})$}.\footnote{More precisely, our running time is $\widetilde\Order(n^{3-3\gamma})$ where $\gamma$ is the solution to the equation $3 - 4\gamma = \omega(1 + 2\gamma, 1, 1)$.}
\end{restatable}

Note that the first bound in the statement, namely $n^{(3+\omega)/2}$, is already $n^{2.5}$ when $\omega=2$ and is therefore sufficient for resolving Open Question~\ref{oq1}.
However, the next-order question after establishing that a problem is intermediate is whether we can get closer to the $n^{2.5}$ bound with \emph{current} techniques; typically, such improvements involve fast \emph{rectangular} matrix multiplication.
Interestingly, for some intermediate problems such as unweighted directed APSP ($O(n^{2.528})$ by~\cite{Zwick02}), Min-Witness Product ($O(n^{2.528})$ by~\cite{KowalukL05}) and Dominance Product ($O(n^{2.659})$ by~\cite{Yuster09}) this has been done, while others such as Bounded Monotone Min-Plus Product~\cite{ChiDX022} are still stuck at the~\makebox{$n^{(3+\omega)/2} = n^{2.686}$} bound.
The second bound in the statement, while giving a small quantitative improvement, qualitatively shows that Node-Weighted APSP is among the easier problems in the intermediate class. It remains to be seen whether further improvements, say to~$O(n^{2.528})$ are possible.

As we discuss in the technical overview (Section~\ref{sec:overview}), the new algorithm for Node-Weighted APSP only uses elementary techniques.
It generalizes in a straightforward way to get a new bound for $d$-Weights APSP.

\begin{restatable}{theorem}{thmapspfewweightsnaive} \label{thm:apsp-few-weights-naive}
There is a deterministic algorithm for $d$-weights APSP running in time~\smash{$\widetilde\Order(d \cdot n^{(3+\omega)/2})$}.\footnote{In fact, using rectangular matrix multiplication the running time can be improved to time~\smash{$\widetilde\Order(n^{3-3\gamma})$} whenever $d = \Theta(n^\delta)$ for some $0 \leq \delta \leq 1$ and where $\gamma$ is the solution to the equation $3 - 4\gamma = \omega(1 + \delta + 2\gamma, 1, 1 + \delta)$.}
\end{restatable}

Notably, while this bound improves on previous work for a range of values of $d$, it still does not make any real progress towards resolving Open Question~\ref{oq2}: if $\omega=2$, our $d\cdot n^{2.5}$ running time becomes cubic when the number of distinct weights is $d \geq \sqrt{n}$ which is exactly the same barrier reached by Yuster's algorithm.

Our next theorem, which is the main technical achievement of this paper, uses additive combinatorics machinery in order to push the threshold all the way to $d=n^{1-\eps}$, if $\omega=2$, positively resolving Open Question~\ref{oq2}.
Consequently, the hard instances of APSP must have $d=n^{1-o(1)}$. In other words, if we could find a method to reduce the number of distinct weights per node by a polynomial factor we would effectively refute the APSP hypothesis.
 
\begin{restatable}[Few-Weights APSP]{theorem}{thmapspfewweights} \label{thm:apsp-few-weights}
For every $\delta > 0$, there is some $\epsilon > 0$ such that $n^{(3-\omega)-\delta}$-weights APSP can be solved in time~\smash{$\Order(n^{3-\epsilon})$} by a deterministic algorithm.
\end{restatable}

With current values of $\omega \le 2.371339$ \cite{AlmanDVXXZ25}, our results push the capabilities of subcubic time algorithms from Yuster's $d \leq n^{0.384}$ \cite{Yuster09} to $d \leq n^{0.628}$. However, the quantitative improvement in our running time, which is hidden by the dependence of $\delta$ on $\eps$ in the statement, is rather modest, as is common when using the Balog--Szemer\'edi--Gowers (BSG) Theorem. Whether our techniques can be substituted by elementary ones, as has been done in some of the other applications of the BSG Theorem in algorithm design~\cite{BringmannGSW16,ChiDX022,KasliwalPS25}, is an interesting open question.

\paragraph{Exact Triangle.}
Our work is inspired by the seminal work of Chan and Lewenstein from STOC 2015 \cite{ChanL15} introducing the BSG Theorem into algorithm design.
The applications in their paper revolved around the 3SUM\footnote{Given $n$ numbers, decide if there are three that sum to zero.} and the Min-Plus Convolution\footnote{Given arrays $a,b$ of length $n$, compute the array $c[i]=\min_{k} a[k]+b[i-k]$.} problems, both conjectured to require quadratic time.
There are three main results in their paper, each achieving truly subquadratic time in a restricted setting: (1) Min-Plus Convolution assuming that the input arrays consist of monotone integers bounded by $O(n)$ (or have the ``bounded difference'' property), (2) 3SUM under the assumption that we have quadratic time to preprocess the universe of numbers, and (3) 3SUM under the assumption that the numbers can be clustered into a sublinear number of short intervals.
Notably, follow-up work has managed to (in a sense) generalize the first two results to the setting of matrices or graphs, solving variants of Min-Plus Matrix Multiplication\footnote{Given $n \times n$ matrices $A, B$, compute the matrix $C[i,j]=\min_k A[i,k] + B[k,j]$.} and the Exact Triangle\footnote{Given $n \times n$ matrices $A, B, C$, decide if there are $i, k, j$ such that $A[i, k] + B[k, j] = C[i, j]$ (an \emph{exact triangle}). In the \emph{All-Edges} Exact Triangle the task is to report for each pair $(i, j)$ if is involved in an exact triangle. See \cref{def:exact-tri}.} problem, but not the third, which may be considered the most complicated result in the original paper.
The first generalization was accomplished by Bringmann, Grandoni, Saha, and Vassilevska Williams \cite{BringmannGSW19} (and improved in \cite{WilliamsX20,ChiDX22,ChiDX022,Durr23}), and the second by Chan, Vassilevska Williams, and Xu \cite{ChanWX23}.

The next theorem presents a subcubic time algorithm for the Exact Triangle problem in graphs with few weights per node. 
Note that in the unrestricted setting, Exact Triangle is a harder problem than both APSP and 3SUM in the sense that a subcubic time algorithm for it breaks the two hypotheses \cite{WilliamsW13,WilliamsW18,Vassilevska18}. 
This result can informally be viewed as a generalization of Chan and Lewenstein's third result on Clustered 3SUM from arrays to matrices or graphs.\footnote{Our Few-Weights Exact Triangle problem can be viewed as an analogue of the \emph{Few-Weights 3SUM Convolution} problem defined as follows: given three length-$n$ arrays $a,b,c$ where the number of distinct weights in $a$ is at most $n^{1-\delta}$, find a pair of indices $i,j$ such that $c[i+j]=a[i]+b[j]$. Consider the standard reduction from this problem to the 3SUM problem by mapping each $a[i]$ to the number $M\cdot a[i] + i$ for a large enough $M$ (and mapping $b[j],c[k]$ similarly): note that the numbers $M\cdot a[i] + i$ can be clustered into $n^{1-\delta}$ intervals $[M\cdot a^*+1,M\cdot a^*+n]$ of length $n$, so this is a Clustered 3SUM instance which is solved by \cite{ChanL15} in \smash{$\widetilde O(n^{2-\delta/7})$} time.}

\begin{restatable}[Few-Weights Exact Triangle]{theorem}{thmexacttrifewweights} \label{thm:exact-tri-few-weights}
For every $\delta > 0$ there is some $\epsilon > 0$ such that $n^{(3-\omega)-\delta}$-weights All-Edges Exact Triangle can be solved in time $\Order(n^{3-\epsilon})$. 
\end{restatable}

Note that this theorem resolves Open Question~\ref{oq2} for Exact Triangle as well. Interestingly, Open Question~\ref{oq1} is not relevant for Exact-Triangle because the node-weighted case has long been known to be just as easy as the unweighted setting \cite{WilliamsW13} and can be solved in $n^{\omega+o(1)}$ time.

\subsection{Hardness and Reductions}
\label{sec:undir}

While the above presentation focused on the directed case, it is important to note that this work resolves Open Questions~\ref{oq1} and~\ref{oq2} in the (easier) undirected case as well. And even in the undirected case previous work by Chan, Vassilevska W., and Xu proved that there can be no algorithm faster than $n^{2.5-o(1)}$, based on the assumption that \emph{directed} APSP in unweighted graphs requires time $n^{2.5-o(1)}$~\cite{ChanWX21} (see the discussion in Section~\ref{sec:hardness}).

We modify their reduction in a simple way to obtain the following bizarre-looking result for APSP in undirected graphs with $d$ weights per node (i.e., where there is only one distinct weight touching each node). For $d = 1$ the problem essentially coincides with undirected \emph{unweighted} APSP and can be solved in $O(n^{\omega})$ time \cite{AlonGM97,Seidel95}. For $d = 2$, we give an $n^{2.5-o(1)}$ based on the same construction from~\cite{ChanWX21}. For larger $d$ this lower bound increases to $\sqrt d n^{2.5-o(1)}$ as detailed in the following theorem.

\begin{restatable}{theorem}{thmdaspsp} \label{thm:dapsp}
Under the Bounded Min-Plus Hypothesis, for every $d=n^\gamma\geq 2$ with $\gamma\in [0,1]$, APSP in (directed or undirected) graphs with at most $d$ distinct weights and hence also $d$-weights APSP cannot be solved in time $O(\sqrt{d}n^{2.5-\epsilon})$, for any constant $\epsilon > 0$.
\end{restatable}

The hypothesis in the statement is a variant of the Strong APSP Hypothesis and is discussed, along with the proof of \cref{thm:dapsp}, in Section~\ref{sec:hardness}.

Note that the conditional lower bound grows to $n^{2.5+\delta}$ if the number of distinct weights grows to $n^{\eps}$. This rules out the possibility of an $n^{2.5}$ algorithm for all sublinear $d$.

Finally, in Section~\ref{sec:hardness}, we build on the techniques in our $d$-Weight APSP algorithm (Theorem~\ref{thm:apsp-few-weights}) in order to show a reduction from a variant of $d$-Weight APSP to Node-Weighted APSP. While being quite involved, this reduction does not use the BSG Theorem, and one may view it as vague evidence that further improvements to Node-Weighted APSP (getting closer to the $n^{2.5}$ bound when $\omega>2$) may come from incorporating our additive combinatorics machinery.
A natural upper bound to aim for is the $O(n^{2+\mu})=O(n^{2.528})$ bound achieved by Zwick's algorithm~\cite{Zwick02} for unweighted directed APSP.\footnote{Here, $0.5 \leq \mu \leq 0.527500$~\cite{AlmanDVXXZ25} is the solution to the equation $1 + 2\mu = \omega(1, \mu, 1)$.}
\section{Preliminaries} \label{sec:prelims}
We set $[n] = \set{1, \dots, n}$ and write $\poly(n) = n^{\Order(1)}$ and $\widetilde\Order(n) = n (\log n)^{\Order(1)}$.

\paragraph{All-Pairs Shortest Paths.}
Let $G = (V, E)$ be a directed graph with node weights or edge weights, and consider a path $P = (v_0, \dots, v_\ell)$ in $G$. We say that $P$ has \emph{hop-length} $\ell$. If $G$ is node-weighted with~\makebox{$w : V \to \Int$}, then we define the \emph{weight} of $P$ by $w(P) = w(v_1) + \dots + w(v_\ell)$ (i.e., we exclude the first vertex in the weight; this ensures that when concatenating two paths their weights add up). If $G$ is edge-weighted, then we define~$w(P)$ in the usual way. We write $D_G \in (\Int \cup \set{-\infty, \infty})^{V \times V}$ for the distance matrix of $G$, i.e., $D_G[u, v]$ is the length of a shortest $u$-$v$-path. The entry is $\infty$ if $v$ cannot be reached from $u$ and the entry is $-\infty$ in the exceptional case that there is a $u$-$v$-path hitting a negative-weight cycle. Similarly, let~\smash{$D_G^{\leq h} \in (\Int \cup \set{\infty})^{V \times V}$} denote the $h$-hop bounded distances, i.e., $D_G^{\leq h}[u, v]$ is the length of a shortest $h$-hop-bounded $u$-$v$-path. We drop the subscript whenever $G$ is clear from context. The \emph{Node-Weighted APSP} problem is to compute $D_G$ given a directed node-weighted graph $G$. The \emph{$d$-Weights APSP} problem is to compute~$D_G$ given a directed edge-weighted graph~$G$ with the promise that each node has outgoing edges of at most~$d$ distinct weights.\footnote{Instead, we can equivalently assume that in the graph each node has \emph{incoming} edges of at most~$d$ distinct weights, by computing APSP on the reversed graph (in which the orientations of all edges are flipped).} Throughout, for all APSP and matrix problems we assume that all involved weights are polynomially bounded (i.e., in the range $\set{-n^c, \dots, n^c}$ for some constant $c$).

\paragraph{Matrices and Min-Plus Product.}
We use $\MM(n_1, n_2, n_3)$ to denote the time complexity for multiplying two integer matrices of dimensions $n_1\times n_2$ and $n_2\times n_3$ respectively. We also let $2 \leq \omega \leq 2.371339$ denote the exponent of matrix multiplication~\cite{AlmanDVXXZ25}, i.e., we assume that $\MM(n,n,n) = \Order(n^{\omega})$ and $\MM(n^a, n^b, n^c) = \Order(n^{\omega(a, b, c)})$.\footnote{Strictly speaking, $\omega$ is defined as a limit value and thus we only have the guarantee that for every $\epsilon > 0$ there is a matrix multiplication algorithm in time $\Order(n^{\omega+\epsilon})$. As is common in the literature, we neglect this technicality for the sake of simplicity.} 

Let $A, B \in (\Int \cup \set{\infty})^{n \times n}$. For sets $S, T \subseteq [n]$ we denote by $A[S, T]$ the restriction of $A$ to the rows indexed by $S$ and the columns indexed by $T$. The \emph{min-plus product} (also called \emph{distance product}) $A \star B \in (\Int \cup \set{\infty})^{n \times n}$ of $A$ and $B$ is defined by
\begin{equation*}
    (A \star B)[i, j] = \min_{k \in [n]} (A[i, k] + B[k, j]).
\end{equation*}
The \emph{Min-Plus Product} problem is to compute the min-plus product of two given matrices $A, B$, and in the \emph{$d$-Weights} Min-Plus Product problem we restrict the matrix $B$ to have at most $d$ distinct non-$\infty$ entries in each column.\footnote{Instead, we can equivalently restrict the rows in $A$ to have at most $d$ distinct non-$\infty$ entries in each row. We can then read off $A \star B$ from $(A \star B)^T = B^T \star A^T$.}

It is well-known~\cite{FischerM71} (and easy to verify) that computing min-plus products captures computing distances in graphs via the observation that $D^{\leq h} = D^{\leq 1} \star \dots \star D^{\leq 1}$ (i.e., the $h$-fold min-plus product of the weighted adjacency matrix with itself).

\paragraph{Sumsets.}
Let $X, Y \subseteq \Int$. We define their \emph{sumset} $X + Y = \set{x + y : x \in X,\, y \in Y}$ and $-X = \set{-x : x \in X}$. Moreover, we say that $z \in X + Y$ has \emph{multiplicity} $r_{X + Y}(z) = \abs{\set{(x, y) \in X \times Y : z = x + y}}$, and we write $P_t(X, Y) = \set{z \in X + Y : r_{X + Y}(z) \geq t}$ to denote the set of \emph{$t$-popular} sums.

\section{Technical Overview} \label{sec:overview}

\subsection{Node-Weighted APSP}
To overview our algorithm for Node-Weighted APSP, we start with a quick recap of the previous algorithms due to Chan~\cite{Chan10} and Yuster~\cite{Yuster09}. A central idea in the context of APSP-type problems is to compute the distances via repeated computations of min-plus products, $D^{\leq h} = D^{\leq 1} \star \dots \star D^{\leq 1}$. For instance, computing APSP in edge-weighted graphs reduces via \emph{repeated squaring} $D^{\leq 2h} = D^{\leq h} \star D^{\leq h}$ to computing $\log n$ min-plus products. Chan's algorithm for Node-Weighted APSP~\cite{Chan10} follows a different approach based on two main insights.

First, observe that since the graph is node-weighted, the matrix $D^{\leq 1}$ admits the simple structure that in each column all non-$\bot$ entries are equal. This can be exploited to compute min-plus products $A \star D^{\leq 1}$ and~\makebox{$D^{\leq 1} \star A$} more efficiently than the naive cubic-time algorithm, namely in time $\widetilde\Order(n^{(3+\omega)/2})$; we will refer to this operation as a \emph{Boolean min-plus product} (see \cref{sec:apsp-node-weighted:sec:min-plus}). To exploit this fact, however, we are forced to compute $D^{\leq h} = D^{\leq 1} \star \dots \star D^{\leq 1}$ via~\makebox{$h-1$} sequential computations in time~\smash{$\widetilde\Order(h \cdot n^{(3+\omega)/2})$}, and cannot rely on repeated squaring anymore. In particular, we cannot simply set $h = n$ to compute $D = D^{\leq n}$.

Instead, the second idea is to sample a uniformly random set~$S$ of~\smash{$\widetilde\Order(n / h)$} \emph{pivot} nodes aiming to \emph{shortcut} all shortest paths with hop-length more than $h$. Specifically, we can efficiently compute the inter-pivot distances $D[S, S]$ in time~$\widetilde\Order(n^3 / h)$ by running Dijkstra's algorithm from each pivot. Moreover, with good probability the set $S$ simultaneously hits, for all nodes $u, v$, all length-$h$ segments in a shortest $u$-$v$-path. We can thus decompose each relevant shortest path of hop-length at least $h$ into an initial length-$h$ segment leading to a pivot $s$, followed by a shortest $s$-$t$-path leading to another pivot $t$, followed by a trailing length-$h$ segment. In other words:
\begin{equation} \label{eq:pivot-distances}
    D = \min\Big(\, D^{\leq h},\, D^{\leq h}[V, S] \star D[S, S] \star D^{\leq h}[S, V] \,\Big).
\end{equation}
This expression for $D$ can be computed by $2h$ Boolean min-plus products in time~\smash{$\widetilde\Order(h \cdot n^{(3+\omega)/2})$}. Trading off the two contributions, the total time is~\smash{$\widetilde\Order(n^{(9 + \omega)/4})$}~\cite{Chan10}. Yuster achieved an improvement by rectangular fast matrix multiplication~\cite{Yuster09}, speeding up the computation of the individual Boolean min-plus products (and some further modifications).

Our algorithm builds on the same framework, but adds two new ideas to the mix. All in all, we still achieve a satisfyingly simple and clean algorithm.

\paragraph{Idea 1: Rectangular Boolean Min-Plus Products.}
Our initial hope was that~\eqref{eq:pivot-distances} can be evaluated in such a way that we only need to compute Boolean min-plus products $A \star D^{\leq 1}$ of \emph{rectangular} matrices~$A$; when~$A$ is an~\makebox{$s \times n$} matrix for some~\makebox{$s \ll n$} we would thereby hope to improve the running time of each individual min-plus computation. At first this approach seems quite hopeless, even for computing $D^{\leq h}$ in~\eqref{eq:pivot-distances}. But suppose for now that we would only want to compute the distances $D[V, S']$ for some small set~\makebox{$S' \subseteq V$}. Then
\begin{equation*}
    D[V, S'] = \min\Big(\, D^{\leq h}[V, S'],\, D^{\leq h}[V, S] \star D[S, S] \star D^{\leq h}[S, S'] \,\Big),
\end{equation*}
and this expression can indeed be evaluated by only rectangular Boolean min-plus products. For the left side in the minimum compute $D^{\leq h}[V, S'] = D^{\leq 1} \star \dots \star D^{\leq 1}[V, S']$ (sweeping from right to left). For the right side we first compute $M := D[S, S] \star D^{\leq 1}[S, V] \star \dots \star D^{\leq 1}$ (sweeping from left to right), and then compute~\makebox{$D[V, S'] = D^{\leq 1} \star \dots \star D^{\leq 1}[V, S] \star M[S, S']$} (sweeping from right to left). All in all, we compute $3h$ Boolean min-plus products, each of size $n \times n \times |S'|$ or $|S| \times n \times n$. This is promising, but to properly exploit this insight we need to restructure the algorithm in \emph{levels.}

\paragraph{Idea 2: Multi-Level Pivots.}
Inspired by Zwick's algorithm for directed unweighted APSP~\cite{Zwick02}, instead of considering one single set of pivots that hits paths of one specific length, consider multiple \emph{levels} of pivots $S_0, \dots, S_L$. Here, $S_\ell$ is a uniform sample of $\widetilde\Order(n / 2^\ell)$ nodes with the purpose to hit paths of length~$2^\ell$, and at the top level we choose $S_0 = V$. The idea is to compute distances~\makebox{$D[S_L, S_L], \dots, D[S_0, S_0]$} step-by-step, from bottom to top. At the base level $\ell = L$ we compute the distances $D[S_L, S_L]$ by $|S_L|$ calls to Dijkstra's algorithm. And at each level $\ell < L$ our goal is to compute~$D[S_\ell, S_\ell]$ having access to the previously computed distances $D[S_{\ell+1}, S_{\ell+1}]$. In this setting we can perfectly apply our previous idea. Specifically, apply the previous paragraph with~\makebox{$S = S_{\ell+1}$, $S' = S_\ell$} and~\makebox{$h = 2^\ell$}; then computing $D[S_\ell, S_\ell]$ amounts to computing~$3 \cdot 2^\ell$ Boolean min-plus products of size~\makebox{$|S_{\ell+1}| \times n \times n$} or~\makebox{$n \times n \times |S_\ell|$} (where both~\smash{$|S_\ell|, |S_{\ell+1}| \leq \widetilde\Order(n / 2^\ell)$}). Choosing appropriate parameters, the running time of this algorithm turns out to be~\smash{$\widetilde\Order(n^{(3+\omega)/2})$}.

\medskip
In \cref{sec:apsp-node-weighted} we elaborate on this algorithm in detail, and also (1)~give a slightly optimized algorithm based on rectangular fast matrix multiplication, (2)~give a derandomization based on Zwick's~\cite{Zwick02} bridging sets (similar to the derandomizations by Chan~\cite{Chan10} and Yuster~\cite{Yuster09}), and (3)~comment that our algorithm easily extends to \emph{negative} and/or \emph{real} weights.

\subsection{Few-Weights APSP}
This algorithm easily extends to an algorithm for $d$-Weights APSP in time $\widetilde\Order(d \cdot n^{(3+\omega)/2})$ (\cref{thm:apsp-few-weights-naive}), with a slight improvement by rectangular matrix multiplication. The only difference is that we replace all Boolean (aka $1$-Weights) Min-Plus Products by $d$-Weights Min-Plus Products, computed naively with an overhead of~$d$. Recall that this leads to a subcubic algorithm whenever $d \ll n^{(3-\omega)/2}$ (i.e., whenever~\makebox{$d \ll \sqrt n$} assuming~\makebox{$\omega = 2$}). However, to achieve our subcubic algorithms for larger values of $d$ up to $n^{3-\omega}$ (\cref{thm:apsp-few-weights}) it seems that we have to take the \emph{structure} of the weights into account. The remainder of this overview is devoted to a proof sketch of \cref{thm:apsp-few-weights} which, compared to the previous approach, turns out to be considerably more technically involved.

\paragraph{Reduction to Exact Triangle.}
In fact, as mentioned before we prove a slightly stronger statement concerning the following \emph{All-Edges Exact Triangle} problem: Given three matrices $A, B, C \in (\Int \cup \set{\bot})^{n \times n}$ (which we typically view as the adjacency matrices of a 3-partite graph, where $\bot$ symbolizes a non-edge), the goal is to decide for each pair $(i, j) \in [n]^2$ if there is some~$k \in [n]$ with
\begin{equation*}
    A[i, k] + B[k, j] = C[i, j];
\end{equation*}
in this case we call the triple $(i, k, j)$ an \emph{exact triangle}. In the modern fine-grained complexity literature it is well-known that the (All-Edges) Exact Triangle is both APSP- and 3SUM-hard~\cite{WilliamsW18,WilliamsW13} (in the sense that a truly subcubic $n^{3-\Omega(1)}$-time algorithm for Exact Triangle implies a truly subcubic $n^{3-\Omega(1)}$-time algorithm for APSP and a truly subquadratic $n^{2-\Omega(1)}$-time algorithm for 3SUM). Adapting these known reductions (\cref{lem:apsp-to-min-plus,lem:min-plus-to-exact-tri}), it similarly follows that in order to achieve an $n^{3-\Omega(1)}$-time algorithm for \emph{$d$-Weights} APSP it suffices to design an $n^{3-\Omega(1)}$-time algorithm for the \emph{$d$-Weights} All-Edges Exact Triangle problem, where we restrict the matrix $A$ to contain at most $d$ distinct entries per row.\footnote{This choice is somewhat arbitrary; we could have similarly restricted the columns in $A$, or the matrices $B$ or $C$ instead of~$A$; see the formal \cref{def:d-weights-exact-tri}.}

Even though we shifted our focus to a harder problem, the advantage of studying All-Edges Exact Triangle rather than APSP directly is that dealing with the ``equality constraint'' rather than the ``minimization constraint'' leads to a more streamlined algorithm.  

\paragraph{Starting Point: BSG Covering.}
The few-weights Exact Triangle problem can be understood as a graph analogue of the \emph{Clustered 3SUM} problem defined by Chan and Lewenstein~\cite{ChanL15}. Therefore, a natural starting point for us is to draw inspiration from their algorithmic framework centered on the Balog-Szemerédi-Gowers (BSG) Theorem~\cite{BalogS94,Gowers01} from additive combinatorics. In fact, all results in their breakthrough paper trace back to the following ``Covering'' version of the BSG Theorem:

\begin{restatable}[BSG Covering~\cite{ChanL15}]{theorem}{thmbsgcovering} \label{thm:bsg-covering}
Let $X, Y, Z \subseteq \Int$ be sets of size at most $d$, and let $K \geq 1$. Then there are subsets $X_1, \dots, X_K \subseteq X$ and $Y_1, \dots, Y_K \subseteq Y$, and a set of pairs $R \subset X \times Y$ satisfying the following three properties:
\begin{enumerate}[label=(\roman*)]
    \item $\set{(x, y) \in X \times Y : x + y \in Z} \subseteq \bigcup_{k=1}^K (X_k \times Y_k) \cup R$.
    \item $|X_k + Y_k| \leq \Order(K^5 d)$ for all $k \in [K]$.
    \item $|R| \leq \Order(d^2 / K)$.
\end{enumerate}
The sets $X_1, \dots, X_K, Y_1, \dots, Y_K, R$ can be computed in deterministic time $\widetilde\Order(d^\omega K)$.
\end{restatable}

Throughout think of $K = d^\epsilon$ for some small constant $\epsilon > 0$. Intuitively, the BSG Covering Theorem states that the sets~$X$ and~$Y$ can be decomposed into an \emph{additively structured} part (the sets $X_1, \dots, X_K, Y_1, \dots, Y_K$) plus a small \emph{remainder} part (the set~$R$). Here, additively structured means that for each pair $X_k, Y_k$ the \emph{sumset}~\makebox{$X_k + Y_k := \set{x + y : x \in X_k, y \in Y_k}$} has small size. It is instructive to think of such highly structured sets as intervals or arithmetic progressions (with the same step width).

It is natural that this theorem finds applications for arithmetic computational problems like 3SUM, but our challenge in the following is to apply this theorem in a \emph{graph} setting.

\paragraph{Step 1: Uniform Regular Exact Triangle via BSG Covering.}
Our strategy is as follows: We will first assume that the graph is sufficiently \emph{structured} (specifically, \emph{$d$-uniform} and \emph{$\frac{n}{d}$-regular} defined as follows) so that we can conveniently exploit the BSG Covering. The main bulk of our work is then to later \emph{enforce} this structure:
\begin{itemize}
    \item \emph{$d$-Uniform: The three matrices $A, B, C$ contain at most $d$ distinct entries.}\newline
    (That is, we not only require that each row of $A$ has at most $d$ distinct entries, but that the \emph{entire} matrix $A$ contains at most $d$ distinct entries, and moreover that the same applies to $B$ and $C$. See \cref{def:d-uniform}.)
    \item \emph{$\frac{n}{d}$-Regular: Each entry appears at most $\frac{n}{d}$ times in its respective row and column in $A$, $B$ and $C$.}\newline
    (See \cref{def:r-regular}.)
\end{itemize}

\begin{restatable}[Uniform Regular Exact Triangle]{lemma}{lemexacttriuniformregular} \label{lem:exact-tri-uniform-regular}
There is a deterministic algorithm solving All-Edges Exact Triangle on $d$-uniform and $\frac{n}{d}$-regular instances in time \smash{$\widetilde\Order(n^{3-(3-\omega)/7} d^{1/7})$}.
\end{restatable}

We give a quick proof sketch of \cref{lem:exact-tri-uniform-regular}. Let $X$ denote the set of non-$\bot$ entries in $A$, let $Y$ denote the set of non-$\bot$ entries in $B$, and let $Z$ denote the set of non-$\bot$ entries in $C$. By the $d$-uniformness assumption we are guaranteed that $|X|, |Y|, |Z| \leq d$. We apply the BSG Covering Theorem to the sets $X, Y, Z$ to compute subsets $X_1, \dots, X_K, Y_1, \dots, Y_K$ and a remainder set $R$. We are left with only two subtasks: Detecting the exact triangles $(i, k, j)$ for which~\makebox{$(A[i, k], B[k, j]) \in X_\ell \times Y_\ell$} for some $\ell \in [K]$ (the additively structured case), and detecting the exact triangles for which~\makebox{$(A[i, k], B[k, j]) \in R$} (the remainder case).

For the additively structured case we enumerate each $\ell \in [K]$ and consider the matrix $A'$ obtained from~$A$ by deleting all entries not in $X_\ell$ (i.e., by replacing these entries by $\bot$) and the analogous matrix $B'$. We then solve the Exact Triangle instance $(A', B', C)$ by an \emph{algebraic} algorithm running in time~\smash{$\widetilde\Order(n^\omega \cdot |X_\ell + Y_\ell|)$}. This algorithm generalizes the~\smash{$\widetilde\Order(n^\omega \cdot M)$}-time algorithm for the $M$-bounded case (which is based on combining fast matrix multiplication with the fast Fourier transform) by incorporating \emph{sparse convolutions} techniques. By the BSG Covering theorem, this case takes time~\smash{$\widetilde\Order(K \cdot n^\omega \cdot d \cdot K^5)$}.

In contrast, we deal with the remainder $R$ by brute-force and simply enumerate all triples $(i, k, j)$ with $(A[i, k], B[k, j]) \in R$. This is where we critically exploit the regularity assumption: For each pair $(a, b) \in R$ and each $k \in [n]$, there can be at most $\frac{n}{d}$ entries $A[i, k] = a$ (using that each entry appears at most $\frac{n}{d}$ times in its respective column in $A$) and at most $\frac{n}{d}$ entries $B[k, j] = b$ (using that each entry appears at most $\frac{n}{d}$ times in its respective row in $B$). Thus, we enumerate at most $|R| \cdot n \cdot \frac{n}{d} \cdot \frac{n}{d} = \Order(n^3 / K)$ triples in total. Here we used that $|R| = \Order(d^2 / K)$ by \cref{thm:bsg-covering}. The lemma follows by trading this off against the running time of the structured case.

\medskip
Until this point we have vaguely followed Chan and Lewenstein's~\cite{ChanL15} approach ported to the Exact Triangle problem. We needed that the instance is $d$-uniform to be able to apply the BSG Covering Theorem in the first place, and we needed that the instance is $\frac{n}{d}$-regular to deal with the unstructured remainder.\footnote{Strictly speaking, we only need the weaker assumption that the columns in $A$ and the rows in $B$ are regular; however, in our later regularization step we get the other regularity assumptions essentially for free.} In our view the $d$-uniform assumption seems fundamental while the $\frac{n}{d}$-regular assumption feels more like a technical detail. Unfortunately, \emph{both} assumptions turned out to be serious barriers which require a lot of work to be established.

\paragraph{Step 2: Uniformization.}
In a first step we will transform an arbitrary $d$-weights instance of Exact Triangle into a few equivalent \emph{$d$-uniform} instances disregarding the $\frac{n}{d}$-regular requirement for now (see \cref{lem:exact-tri-uniformization}). On a conceptual level, our goal is to transform a given instance with a \emph{local} structure ($d$~distinct weights per node) to instances with \emph{global} structure ($d$~distinct weights in total). We will achieve this by exploiting the \emph{all-pairs} nature of the problem: While the weights of any two specific nodes $i$ and $j$ can be adversarially uncorrelated, this cannot simultaneously apply to all pairs of nodes $i, j$ without rendering the instance easy to solve.

To illustrate the concrete idea, consider a pair $(i, j)$ and let $X_i$ denote the set of entries in the $i$-th row~of~$A$ (i.e., the set of edge weights incident to $i$), and similarly let $Y_j$ denote the set of entries in the $j$-th column of $B$. The $d$-weights assumption implies that $|X_i| \leq d$ for all $i$, and by a simple trick which we omit here one can additionally assume that $|Y_j| \leq d$. We say that an element $c \in X_i + Y_j$ is \emph{popular} if there are many representations $c = a + b$ for $(a, b) \in X_i \times Y_j$, and \emph{unpopular} otherwise. (As the precise threshold we will later pick $\Theta(d / \Delta)$ where $\Delta = d^\epsilon$ for some small constant $\epsilon > 0$.) In order to test whether $(i, j)$ is contained in a exact triangle, it would be very helpful if $c := C[i, j]$ was \emph{unpopular}. In this case we could simply enumerate all possible representations~\makebox{$c = a + b \in X_i + Y_j$}, and test whether there is an intermediate node~$k$ with~$A[i, k] = a$ and~\makebox{$B[k, j] = b$} by brute-force. The problematic case is when $c$ is \emph{popular}. Perhaps a first instinct is that having a popular sum implies that $X_i$ and $Y_j$ must be additively structured sets, such as intervals or arithmetic progressions. However, this is not necessarily true as the following example shows: Let $X_i$ be a \emph{random} set and choose $Y_j = -X_i := \set{-a : a \in X_i}$. Then~\makebox{$0 \in X_i + Y_j$} has the highest-possible multiplicity $|X_i|$ in the sumset. However, this example also hints at what general structure we can expect: Generally, it is easy to see that there must be large subsets $S \subseteq X_i$ and~\makebox{$T \subseteq Y_j$} such that~\makebox{$S = -T$}. We will now exploit this structure as follows: Consider the bipartite graph~\makebox{$H = ([n], [n], E)$} where we include an edge $(i, j) \in E$ if and only if the sumset $X_i + Y_j$ has a popular element. That is, $E$ contains exactly the problematic pairs. Then either $H$ is sparse (i.e., there is a subquadratic number of problematic pairs $(i, j)$ leading to a simple subcubic algorithm for Exact Triangle), or $E$ is dense and thus there should be sets $S$ and $T$ that constitute a common core not only to a specific pair $(i, j)$, but to \emph{many} pairs $(i, j)$. Formally, this idea leads to the following result which we call the \emph{``popular sum decomposition'':}
\begin{restatable}[Popular Sum Decomposition]{lemma}{lempopsumsdecomp} \label{lem:pop-sums-decomp}
Let $\Delta \geq 1$ and let $X_1, \dots, X_n, Y_1, \dots, Y_n \subseteq \Int$ be sets of size at most $d$. Then there are partitions
\begin{alignat*}{2}
    X_i &= X_{i, 1} \sqcup \dots \sqcup X_{i, \Delta^2} \sqcup X'_i \qquad &&\text{(for all $i \in [n]$),} \\
    Y_j &= Y_{j, 1} \sqcup \dots \sqcup Y_{j, \Delta^2} \sqcup Y'_j \qquad &&\text{(for all $j \in [n]$)}
\end{alignat*}
with the following properties:
\begin{enumerate}[label=(\arabic*)]
    \item There are sets $S_\ell$ and shifts $s_{i, \ell}$ (for $i \in [n]$ and $\ell \in [\Delta^2]$) such that $X_{i, \ell} \subseteq s_{i, \ell} + S_\ell$ and $|S_\ell| \leq d$, and\newline there are sets $T_\ell$ and shifts $t_{j, \ell}$ (for $j \in [n]$ and $\ell \in [\Delta^2]$) such that $Y_{j, \ell} \subseteq t_{j, \ell} + T_\ell$ and $|T_\ell| \leq d$.
    \item $\abs{\set{(i, j) \in [n] : P_{2d / \Delta}(X'_i, Y_j) \neq \emptyset}} \leq n^2 / \Delta$, and\newline $\abs{\set{(i, j) \in [n] : P_{2d / \Delta}(X_i, Y'_j) \neq \emptyset}} \leq n^2 / \Delta$.
\end{enumerate}
These partitions can be found by a randomized algorithm in time~\smash{$\widetilde\Order(n^2 d \cdot \Delta^3)$}, or a deterministic algorithm in time~\smash{$\widetilde\Order(n^2 d \cdot \Delta^3) \cdot U^{o(1)}$}, where $U$ denotes the largest absolute value of any input integer.
\end{restatable}

To apply this decomposition to the Exact Triangle problem, partition $A$ into matrices $A_1, \dots, A_{\Delta^2}$ and~$A'$ where we restrict the entries in the $i$-th row to $X_{i, 1}, \dots, X_{i, \Delta^2}$ and $X_i'$. Similarly partition $B$ based on the decomposition of the $Y_j$'s. Observe that each exact triangle appears in exactly one instance $(A_g, B_h, C)$ (the \emph{ordinary} instances) or~\makebox{$(A_g, B', C)$} or~\makebox{$(A', B_h, C)$} or~\makebox{$(A', B', C)$} (the \emph{exceptional} instances) for~\makebox{$g, h \in [\Delta^2]$}.

To deal with the exceptional instances we exploit that there are only $\Order(n^2 / \Delta)$ problematic pairs by Property~(2) of the decomposition (recall that $P_{2d / \Delta}(X_i', Y_j)$ is the set of popular elements in $X_i' + Y_j$; thus, when this set is empty the pair $(i, j)$ is unproblematic).

For the ordinary instances $(A_g, B_h, C)$ consider the following transformation: From each entry $A_g[i, k]$ we subtract $s_{i, g}$, from each entry $B_h[k, j]$ we subtract $t_{j, h}$, and from each entry $C[i, j]$ we subtract $s_{i, g} + t_{j, h}$. Clearly, this transformation does not change the set of exact triangles (as for each triple $(i, k, j)$ we subtract the same number on both sides of $A_g[i, k] + B_h[k, j] = C[i, j]$). However, by Property~(1) this transformation leads to matrices~$A^*, B^*, C^*$ such that all entries of $A^*$ are contained in $S_g$ and all entries of $B^*$ are contained in $T_h$. In particular, the resulting matrices $A^*$ and $B^*$ are $d$-uniform.

This proof sketch explains our main innovations, but the formal proof requires some more technical ideas such as (1) a ``mini-regularization'' preprocessing step to deal with the unpopular pairs, (2) an efficient algorithm to compute the popular sums~\cite{FischerJX25}, and (3) the insight that finally $C^*$ can also be easily uniformized. We defer these details to \cref{sec:apsp-few-weights:sec:uniformization}.

\paragraph{Step 3: Regularization.}
The final reduction step is to establish the $\frac{n}{d}$-regular assumption (\cref{lem:exact-tri-regularization}). Let us emphasize again that, despite our first impression that this issue should be easily solvable by standard techniques, it turned out to be a serious technical issue. In fact, as we will describe shortly, the regularization step turns out to be the most costly step of our entire algorithm.

From the previous step we assume the following primitive: In time $\Order(n^3 / \Delta)$ we can transform a $d$-weights Exact Triangle instance $(A, B, C)$ equivalently into, say, $\Delta^{100}$ many $d$-uniform Exact Triangle instances. Recall that in a $d$-weights instance we restrict the rows of $A$ to have at most $d$ distinct entries. However, an important insight at this point is that it does not really matter which matrix $A$, $B$ or $C$ is $d$-weights and whether we restrict the rows or columns---we can always appropriately rotate and/or transpose the instance.

The idea behind our regularization step can be sketched as follows. For simplicity of exposition here we only focus on the regularization of, say, the rows of $B$; all five other requirements work similarly. Let~$(A, B, C)$ be a given $d$-weights instance. We first transform it into $\Delta^{100}$ uniform instances, and focus on one resulting instance $(A', B', C')$. We delete all entries in $B'$ that appear at least $\frac{n}{d} \cdot \rho$ times in their respective row (for some parameter $\rho \geq 1$), planning to deal with these deleted entries later by recursion. The remaining matrix can be further partitioned into $\rho$ matrices $B_1', \dots, B_\rho'$ such that each of these matrices is $\frac{n}{d}$-column-regular (and still $d$-uniform, of course). In particular, the instances $(A', B'_i, C')$ are exactly as desired. But we still have to deal with the deleted entries: Let $B''$ denote the matrix consisting of the entries deleted earlier. The key idea is that each entry in $B''$ appears at least $\frac{n}{d} \cdot \rho$ times in its respective row, and thus each row contains at most $d / \rho$ distinct entries. Thus, $(A', B'', C')$ is a $(d / \rho)$-weights Exact Triangle instance, and it is reasonable to solve this instance recursively. This leads to a recursive algorithm with recursion depth $\log_\rho(d)$.

However, there is a serious problem: Already at the first recursive level the decomposition algorithm becomes too inefficient: Recall that we have created $\Delta^{100}$ recursive calls---one for each uniform piece. But each such recursive call first applies the uniformization again which each takes time $\Order(n^3 / \Delta)$, and thus takes super-cubic time $\Order(n^3 \cdot \Delta^{99})$ in total. Our solution to this problem is to carefully set the parameters~$\Delta$ (and~$\rho$) \emph{depending on the recursion depth.} For instance, at the first level of the recursion we would set~\makebox{$\Delta_1 = \Delta^{101}$} so that the total running time of the first level becomes $\Order(n^3 / \Delta)$.

All in all, this recursive approach leads to an algorithm that in time roughly~\smash{$\widetilde\Order(n^3 / \Delta)$} transforms a given $d$-weights Exact Triangle instance into~\smash{$\widetilde\Order(d^{\epsilon} \Delta^{2^{\Order(1/\epsilon)}})$} equivalent $d$-uniform and $\frac{n}{d}$-regular instances, where the constant $\epsilon > 0$ can be chosen arbitrarily small. This is indeed sufficient to prove \cref{thm:apsp-few-weights}: For any~\makebox{$\delta > 0$}, we can solve $n^{(3-\omega)-\delta}$-Weights APSP in subcubic time by calling the regularization decomposition with parameters~\makebox{$\epsilon = \Theta(\delta)$} and~\smash{$\Delta = n^{2^{-\Order(1/\delta)}}$}. However, the resulting subcubic running time is~\smash{$\Order(n^{3-2^{-\Order(1/\delta)}})$}, so the savings have decreased \emph{exponentially.} With more technical effort we can improve this dependence to polynomial (see \cref{rem:eps-dependence}). We leave it as an open question whether the dependence can be improved to linear, i.e., whether there is an algorithm for $d$-Weights APSP running in time, say, $\Order(n^{3-0.001} d^{0.001})$.
\section{Node-Weighted APSP} \label{sec:apsp-node-weighted}
In this section we give the details for our algorithm for Node-Weighted APSP (\cref{thm:apsp-node-weighted}). This algorithm readily generalizes to an algorithm for $d$-Weights APSP (\cref{thm:apsp-few-weights-naive})---however, for the sake of a simpler presentation we mainly focus on Node-Weighted APSP here and only later, in \cref{sec:apsp-node-weighted:sec:d-weights}, consider the $d$-weights variant.

\subsection{Boolean Min-Plus Product} \label{sec:apsp-node-weighted:sec:min-plus}
Let $A, B$ be integer matrices. We defined their \emph{min-plus product $A \star B$} by
\begin{equation*}
    (A \star B)[i, j] = \min_{k \in [n]} (A[i, k] + B[k, j]).
\end{equation*}
It is well-known in the fine-grained complexity literature that the APSP problem (in edge-weighted graphs) is precisely captured by computing min-plus products. To compute APSP in node-weighted graphs, Chan~\cite{Chan10} and Yuster~\cite{Yuster09} relied on a variant to which we will refer as \emph{Boolean min-plus product} defined as follows. Let~$A$ be an integer matrix, and let $B$ be a $0$-$1$-matrix. Then their Boolean min-plus product $A \ostar B$ is defined by
\begin{equation*}
    (A \ostar B)[i, j] = \min\set{A[i, k] : B[k, j] = 1}.
\end{equation*}
We will similarly rely on computing this product, but crucially rely on the following lemma for computing \emph{rectangular} such products. The proof is a simple adaption of the respective lemmas concerned with the square case (see~\cite[Theorem~3.2]{Chan10} or~\cite[Lemma~3.3]{Yuster09}).

\begin{lemma}[Rectangular Boolean Min-Plus Product] \label{lem:min-plus-boolean}
Let $A \in \Int^{s \times n}$ and $B \in \set{0, 1}^{n \times n}$. There is a deterministic algorithm computing $A \ostar B$ in time
\begin{equation*}
    T_{\ostar}(s, n, n) = \Order\parens*{\min_{\Delta \geq 1} \parens*{\MM(s \Delta, n, n) + \frac{s n^2}{\Delta}}}.
\end{equation*}
\end{lemma}
\begin{proof}
For every $i \in [s]$, sort the entries of the $i$-th row of $A$, forming a list $L_i$. Then split $L_i$ (in the sorted order) into $\Delta$ buckets of $\Theta(n/\Delta)$ entries each. Let $L_{i, b}$ be the entries in the $b$-th bucket. Form an $(s\Delta) \times n$ Boolean matrix $A'$ defined as
\begin{equation*}
    A'[(i, b), k] =
    \begin{cases}
        1 &\text{if $A[i, k] \in L_{i, b}$,} \\
        0 &\text{otherwise,}
    \end{cases}
\end{equation*}
and multiply $A'$ by $B$ (over the Boolean semiring) in time $\MM(s\Delta, n, n)$. For every pair $i \in [s], j \in [n]$, find the smallest bucket $b$ such that $(A' B)[(i, b), j] = 1$ and then search through the list $L_{i, b}$ for the smallest entry $A[i, k] \in L_{i, b}$ such that $B[k, j] = 1$. This takes time $O(sn \cdot n/\Delta)$ overall. The correctness should be clear, and the running time is $\Order(\MM(s \Delta, n, n) + sn^2 / \Delta)$ (where the parameter $\Delta \geq 1$ can be freely chosen).
\end{proof}

\begin{lemma} \label{lem:min-plus-node-weighted}
Let $G = (V, E, w)$ be a directed node-weighted graph, let $S \subseteq V$, $A \in (\Int \cup \set{\infty})^{S \times V}$, and~\makebox{$h \geq 1$}. Then we can compute the matrix~\smash{$A \star D_G^{\leq h}$} in deterministic time $\Order(h \cdot T_{\ostar}(|S|, n, n))$. Moreover, in the same time we can compute for each pair $(s, v) \in S \times V$ a \emph{witness} path $v_0, \dots, v_\ell = v$ in $G$ of hop-length $\ell \leq h$ minimizing $A[s, v_0] + w(v_1) + \dots + w(v_\ell)$.
\end{lemma}

To intuitively understand this lemma, first think of the trivial $S \times V$ matrix $A$ with zeros in all positions $A[s, s]$ and $\infty$ elsewhere. In this case the lemma simply computes the distance matrix $A \star D^{\leq h} = D^{\leq h}[S, V]$ (i.e., the lengths of the shortest $h$-hop-bounded paths from all nodes $S$ to all nodes in $V$). In general, the matrix $A$ may store the length of some previously computed paths (from $S$ to $V$); then the lemma computes the distances obtained by following a path represented by $A$ and then taking up to $h$ hops in $G$.

\begin{proof}
Write~\smash{$A_h := A \star D^{\leq h}$}, let~\smash{$B \in \set{0, 1}^{V \times V}$} denote the adjacency matrix of $G$, and let~\smash{$W \in \Int^{S \times V}$} be the matrix defined by $W[u, v] = w(v)$. Observe that:
\begin{align*}
    A_0 &= A, \\
    A_{h+1} &= \min(A_h, (A_h \ostar B) + W).
\end{align*}
For the first identity, note that the $0$-hop distance matrix~\smash{$D^{\leq 0}$} is the matrix with the zeros on the diagonal and $\infty$ elsewhere. For the second identity, focus on any pair of nodes $u \in S, v \in V$, and let~$x \in V$ be a node with~\smash{$A_{h+1}[u, v] = A[u, x] + D^{\leq h+1}[x, v]$}. Let $P$ be a shortest $(h+1)$-hop path from $x$ to $v$ in~$G$. Either~$P$ has hop-length at most~$h$; in this case it follows that~\smash{$A_{h+1}[u, v] = A[u, x] + D^{\leq h}[x, v] \leq A_h[u, v]$}. Or the predecessor $y$ of $v$ in $P$ satisfies that~\smash{$D^{\leq h+1}[x, v] = D^{\leq h}[x, y] + w(v)$} and $B[y, v] = 1$; in this case it follows that~\smash{$A_{h+1}[u, v] = A[u, x] + D^{\leq h}[x, y] + w(v) \leq (A_h \ostar B)[u, v] + w(v)$}. This completes the ``$\leq$'' proof of the second identity, and the other direction is analogous.

From these two expressions it is immediate that we can compute $A_h$ by a sequence of $h$ calls to a Boolean Min-Plus oracle. Each call is with matrices of size $|S| \times n$ by $n \times n$, leading to the claimed time bound $\Order(h \cdot T_{\ostar}(|S|, n, n))$.

Finally, we sketch how to keep track of paths witnessing the computed distance in the same running time. To do this, we need to compute for each computed Boolean min-plus product $A \ostar B$ also some witnesses~$k_{i, j}$ satisfying that $(A \ostar B)[i, j] = A[i, k_{i, j}]$ and $B[k_{i, j}, j] = 1$. To achieve this, simply let $A'$ be the matrix defined by $A'[i, k] = n \cdot A[i, k] + k$ and compute the Boolean min-plus product $A' \ostar B$ instead of $A \ostar B$. We can read off $A \ostar B$ from $A' \ostar B$ by dividing each entry by $n$, and the remainders modulo $n$ are exactly the witnesses~$k_{i, j}$.
\end{proof}

We remark that the previous \cref{lem:min-plus-node-weighted} can similarly compute the matrix~\smash{$D_G^{\leq h} \star A$} by applying the lemma with $A^T$ and with the reverse graph (obtained by reversing the orientations of all edges in $G$), with some minor modifications to take care of the weights associated to the reversed edges.\footnote{More specifically, note that~\smash{$(D_G^{\leq h} \star A)[v,s]$} is the minimum of $w(v_1)+\dots + w(v_\ell)+A[v_\ell,s]$ over all paths $(v_0=v,v_1,\dots,v_\ell)$ of hop-length $\ell\le h$. Let $G^T$ denote the reverse graph of $G$, and let $B[u,s] = w(u)+A[u,s]$. Then $(B^T \star D^{\le h}_{G^T})[s,v] - w(v) = (D_G^{\leq h} \star A)[v,s]$.}

\subsection{A Simple Randomized Algorithm} \label{sec:apsp-node-weighted:sec:rand}
We start with a simple randomized~\smash{$\widetilde\Order(n^{(3+\omega)/2})$}-time algorithm for Node-Weighted APSP. On a high level it resembles previous algorithms for all-pairs problems, notably such as Zwick's algorithm for unweighted directed APSP~\cite{Zwick02}, or the undirected bounded-weight APSP algorithm in~\cite{ChanWX21}.

\begin{lemma} \label{lem:apsp-node-weighted-rand}
Node-Weighted APSP is in randomized time
\begin{equation*}
    \widetilde\Order\parens*{\min_{h \geq 1} \parens*{\frac{n^3}{h^3} + h \cdot T_{\ostar}(n / h, n, n)}}.
\end{equation*}
The algorithm is obtained via an~\smash{$\widetilde\Order(n^3 / h^3)$} time reduction from Node-Weighted APSP to~\smash{$\widetilde\Order(h)$} calls to an oracle for $(n/h) \times n \times n$ Boolean Min-Plus matrix products.
\end{lemma}
\begin{proof}
We start with a description of the algorithm. Let $G = (V, E)$ denote the given graph, let $h \geq 1$ be a parameter, and let $L = \ceil{\log h}$. We construct sets of nodes $S_0, \dots, S_L \subseteq V$ where $S_\ell \subseteq V$ is a uniform sample with rate $\min(100 \log n \cdot 2^{-\ell}, 1)$ (so that in particular, $S_0 = V$). Our strategy is to design an algorithm that runs in \emph{levels $\ell \gets L, L-1, \dots, 0$}, and on the $\ell$-th level computes the distances $D[S_\ell, S_\ell]$. At the final level~\makebox{$\ell = 0$}, we thereby compute the entire distance matrix $D = D[S_0, S_0]$ as required.

\begin{itemize}
    \setlength\parskip{0pt plus 1pt}
    \item \emph{Base level $\ell = L$:} We could compute $D[S_L, S_L]$ by solving Single-Source Shortest Paths for every source in~$S_L$ (by Dijkstra's algorithm, or in the presence of negative weights by the recent near-linear-time algorithm~\cite{BernsteinNW22,BringmannCF23}). This is already enough to ultimately obtain the~\smash{$\widetilde\Order(n^{(3+\omega)/2})$}-time algorithm, however, we can do slightly better as follows.
    
    First compute the distance matrix~\smash{$D^{\leq 2^L}[S_L, V]$} by one call to \cref{lem:min-plus-node-weighted}. Then compute~$D[S_L, S_L]$ via repeatedly squaring (i.e., by computing~$\log n$ min-plus products of size $|S_L| \times |S_L| \times |S_L|$) based on the following identity:
    \begin{equation} \label{lem:apsp-node-weighted-rand:eq:base}
        D[S_L, S_L] = \Big(\, D^{\leq 2^L}[S_L, S_L] \,\Big)^{\star n} = \underbrace{D^{\leq 2^L}[S_L, S_L] \star \dots \star D^{\leq 2^L}[S_L, S_L]}_n.
    \end{equation}
    \item \emph{Levels $\ell \gets L-1, \dots, 0$:} We assume that the distance matrix $D[S_{\ell+1}, S_{\ell+1}]$ has already been computed. Then we compute the distance matrix $D[S_\ell, S_\ell]$ with three calls to \cref{lem:min-plus-node-weighted} based on the following identity:
    \begin{equation} \label{lem:apsp-node-weighted-rand:eq:bridging}
        D[S_\ell, S_\ell] = \min\Big(\,\underbrace{D^{\leq 2^\ell}[S_\ell, S_\ell]}_{\text{(i)}},\,\, \underbrace{D^{\leq 2^\ell}[S_\ell, S_{\ell+1}] \star \underbrace{D[S_{\ell+1}, S_{\ell+1}] \star D^{\leq 2^\ell}[S_{\ell+1}, S_\ell]}_{\text{(ii)}}}_{\text{(iii)}}\,\Big).
    \end{equation}
    Specifically, we (i) compute~\smash{$M_1 := D^{\leq 2^\ell}[S_\ell, V] = A \star D^{\leq 2^\ell}$} with a trivial matrix $A$ (the~\makebox{$S_\ell \times V$} matrix with $A[s, s] = 0$ and $\infty$ elsewhere), (ii) compute~\smash{$M_2 := D[S_{\ell+1}, S_{\ell+1}] \star D^{\leq 2^\ell}[S_{\ell+1}, V] = A \star D^{\leq 2^\ell}$} with the matrix~$A$ obtained by extending~\smash{$D[S_{\ell+1}, S_{\ell+1}]$ to an $S_{\ell+1} \times V$ matrix filled up with $\infty$}, and finally (iii) compute~\smash{$M_3 := D^{\leq 2^\ell}[S_\ell, S_{\ell+1}] \star D[S_{\ell+1}, S_{\ell+1}] \star D^{\leq 2^\ell}[S_{\ell+1}, S_\ell] = D^{\leq 2^\ell}[S_\ell, S_{\ell+1}] \star A$} with the matrix~\makebox{$A = M_2[S_{\ell+1}, S_{\ell}]$} (here we apply \cref{lem:min-plus-node-weighted} on the reverse graph). Then the above identity entails that $D[S_\ell, S_\ell] = \min(M_1[S_\ell, S_\ell], M_3)$. This completes the description of the algorithm.
\end{itemize}

\paragraph{Correctness.}
For the correctness we prove that~\eqref{lem:apsp-node-weighted-rand:eq:base} and~\eqref{lem:apsp-node-weighted-rand:eq:bridging} hold with high probability. For~\eqref{lem:apsp-node-weighted-rand:eq:base} we prove that for each pair of nodes $u, v$ there is a shortest $u$-$v$-path $P$ such that no length-$2^L$ subpath of $P$ avoids the set $S_L$. A fixed length-$2^L$ path avoids $S_L$ with probability~\smash{$(1 - 100 \log n / 2^L)^{2^L} \leq \exp(-100 \log n) = n^{-100}$}, so by a union bound over all pairs $u, v$ and all subpaths,~\eqref{lem:apsp-node-weighted-rand:eq:base} holds with probability at least $n^{-97}$.

Next, consider~\eqref{lem:apsp-node-weighted-rand:eq:bridging}. The ``$\leq$'' direction is clear, since any distance in the right-hand side can be realized by a path in~$G$. For the~``$\geq$'' direction take any pair of nodes $u, v \in S_\ell$, and let~$P$ be a shortest $u$-$v$-path. If~$P$ has hop-length at most $2^\ell$, then clearly~\smash{$D[u, v] = D^{\leq 2^\ell}[u, v]$}, so suppose otherwise that~$P$ has hop-length more than~$2^\ell$. We claim that with high probability, both the length-$2^\ell$ prefix and the length-$2^\ell$ suffix of~$P$ contain a vertex in~$S_{\ell+1}$. In that case there clearly are nodes $s, t \in S_{\ell+1}$ so that we can partition~$P$ into three paths~$P_{us}, P_{st}, P_{tv}$, where $P_{us}$ is a shortest $u$-$s$-path with hop-length at~most~$2^{\ell}$,~$P_{st}$~is a shortest $s$-$t$-path (of unbounded hop-length), and $P_{tv}$ is a shortest $t$-$v$-path with hop-length at most $2^{\ell}$, and it follows that:
\begin{align*}
    D[u, v] &= D^{\leq 2^\ell}[u, s] + D[s, t] + D^{\leq 2^\ell}[t, v] \\
    &\leq (D^{\leq 2^\ell}[S_\ell, S_{\ell+1}] \star D[S_{\ell+1}, S_{\ell+1}] \star D^{\leq 2^\ell}[S_{\ell+1}, S_\ell])[u, v].
\end{align*}

To see that this event indeed happens with high probability, recall that $S_{\ell+1}$ is an independent and uniform sample of $V$ with rate $100 \log n \cdot 2^{-\ell-1}$. Thus, the probability of missing all nodes of the length-$2^\ell$ prefix of $P$ is at most~\smash{$(1 - 100 \log n / 2^{\ell+1})^{2^\ell} \leq \exp(-100 \log n \cdot 2^{\ell} / 2^{\ell+1}) \leq n^{-50}$}, and the same bound applies to the suffix. Taking a union bound over all $L \leq \Order(\log n)$ levels, and over the~$n^2$ pairs $(u, v)$, the total error probability becomes $\Order(n^{-47})$. (Of course, the constant can be chosen arbitrarily larger by adjusting the sampling rate.)

\paragraph{Running Time.}
Observe that with high probability it holds that $|S_\ell| = \Order(n \log n / 2^\ell)$, for all levels~$\ell$. In the base case we first call \cref{lem:min-plus-node-weighted} in time~\smash{$\widetilde\Order(h \cdot T_{\ostar}(n/h, n, n))$}, and then compute $\Order(\log n)$ min-plus products of size $|S_L| \times |S_L| \times |S_L|$, taking time~\smash{$\widetilde\Order(|S_L|^3) = \widetilde\Order(n^3 / h^3)$}. In each of the levels $\ell = L-1, \dots, 0$ we call \cref{lem:min-plus-node-weighted} three times, with hop parameter $2^\ell$ and a matrix $A$ of size $S_\ell \times V$ or $S_{\ell+1} \times V$. All in all, the running time becomes
\begin{equation*}
    \widetilde\Order\parens*{\frac{n^3}{h^3} + \sum_{\ell=0}^L 2^\ell \cdot T_{\ostar}(\widetilde\Order(n / 2^{\ell}), n, n)}.
\end{equation*}
Now, for every $n$ and $1 \leq s' \leq s$, we can compute the Boolean Min-Plus product of dimensions $s \times n \times n$ by splitting the $s$ rows of the first matrix into $\Order(s / s')$ parts of size $s'$, thus reducing the problem to $\Order(s / s')$ calls of Boolean Min-Plus product of dimensions $s' \times n \times n$. In other words, $T_{\ostar}(s, n, n) = \Order(s/s' \cdot T_{\ostar}(s', n, n))$. This means that we can simplify the above expression to
\begin{equation*}
    \widetilde\Order\parens*{\frac{n^3}{h^3} + \sum_{\ell=0}^L 2^L \cdot T_{\ostar}(n / 2^L, n, n)} = \widetilde\Order\parens*{\frac{n^3}{h^3} + h \cdot T_{\ostar}(n / h, n, n)}.
\end{equation*}
Another way to say this is that Node-Weighted APSP can be solved in time~\smash{$\widetilde\Order(n^3 / h^3)$} by~\smash{$\widetilde\Order(h)$} calls to an oracle for $(n/h) \times n \times n$ Boolean Min-Plus matrix products.
\end{proof}

We can now complete the randomized part of \cref{thm:apsp-node-weighted}. By combining \cref{lem:min-plus-boolean,lem:apsp-node-weighted-rand}, we obtain that Node-Weighted APSP can be solved in time
\begin{equation*}
    \widetilde\Order\parens*{\frac{n^3}{h^3} + h \cdot \parens*{\MM((n/h) \Delta, n, n) + \frac{(n/h)n^2}{\Delta}}} = \widetilde\Order\parens*{\frac{n^3}{h^3} + h \cdot \MM(n \Delta / h, n, n) + \frac{n^3}{\Delta}},
\end{equation*}
where the parameters $h$ and $\Delta$ can be freely chosen. Picking $h = \Delta = n^{(3-\omega)/2}$, this becomes $\widetilde\Order(n^{(3+\omega)/2})$ (noting that $\MM(n \Delta / h, n, n) = \MM(n, n, n) = \Order(n^\omega)$). However, there is a better choice for the parameters exploiting rectangular matrix multiplication: Let $h = n^\gamma$ and $\Delta = n^{3\gamma}$, where $\gamma$ is the solution to the equation~\makebox{$3 - 3\gamma = \gamma + \omega(1 + 2\gamma, 1, 1)$} (or equivalently, $3 - 4\gamma = \omega(1 + 2\gamma, 1, 1)$). Then the total running time becomes~\smash{$\widetilde\Order(n^{3-3\gamma})$}. The state-of-the-art bounds for fast matrix multiplication due to Alman, Duan, Vassilevska W., Xu, Xu and Zhou~\cite{AlmanDVXXZ25} imply that $\gamma \geq 0.1113$,\footnote{From~\cite{AlmanDVXXZ25} we have that $\omega(1.2,1,1)\leq 2.535063$ and $\omega(1.5,1,1)\leq 2.794633$. Let $\gamma^* = 0.1113318869$, and verify that $1 + 2\gamma^* \leq s \cdot 1.2 + (1-s) \cdot 1.5$ for $s = 0.9244540873$. By the convexity of $\omega(\cdot, \cdot, \cdot)$ it follows that $\omega(1 + 2\gamma^*, 1, 1) \leq s \cdot \omega(1.2, 1, 1) + (1-s) \cdot \omega(1.5, 1, 1) = s \cdot 2.535063 + (1 - s) \cdot 2.794633 \leq 2.5546724526 \leq 3 - 4\gamma^*$. Since in the equation $\omega(1 + 2\gamma, 1, 1) = 3 - 4\gamma$ the left-hand side is nondecreasing with $\gamma$ and the right-hand side is decreasing with $\gamma$, from $\omega(1 + 2\gamma^*, 1, 1) \leq 3 - 4\gamma^*$ it follows that $\gamma \geq \gamma^* = 0.1113318869$.} and so our algorithm runs in time~\smash{$\widetilde\Order(n^{2.6661})$}.

\subsection{Derandomization Using Bridging Sets} \label{sec:apsp-node-weighted:sec:det}
Inspecting the algorithm developed in the last section, observe that the only randomized part is to construct the sets $S_0, \dots, S_L$. We will now replace this randomized construction by a deterministic one based on a variant of Zwick's~\cite{Zwick02} \emph{bridging sets.} In their algorithms, Chan~\cite{Chan10} and Yuster~\cite{Yuster09} exploited the same idea (though the details differ). The high-level idea is to compute hop-bounded $V$-$S_\ell$-distances top-to-bottom (i.e., for $\ell = 0, \dots, L$), and explicitly maintain shortest paths realizing these distances. We can then apply a deterministic Hitting Set algorithm to simultaneously hit all shortest paths from the previous level to assign the current level $S_\ell$. This approach easily achieves a deterministic guarantee for Equation~\eqref{lem:apsp-node-weighted-rand:eq:bridging}. Unfortunately, guaranteeing Equation~\eqref{lem:apsp-node-weighted-rand:eq:base} involves more details:

\begin{lemma} \label{lem:apsp-node-weighted-det}
Node-Weighted APSP is in deterministic time
\begin{equation*}
    \widetilde\Order\parens*{\min_{h \geq 1} \parens*{\frac{n^3}{h^3} + h \cdot T_{\ostar}(n / h, n, n)}}.
\end{equation*}
In other words, the algorithm is obtained via an~\smash{$\widetilde\Order(n^3 / h^3)$}-time reduction from Node-Weighted APSP to~\smash{$\widetilde\Order(h)$} calls to an oracle for $(n / h) \times n \times n$ Boolean Min-Plus products.
\end{lemma}

\begin{proof}
Let $G = (V, E, w)$ be the given node-weighted graph and as before let $L = \ceil{\log h}$. We devise an algorithm running in four steps---in the first three steps we deterministically set up the pivot sets $S_0, \dots, S_L$ and a new set $S^*$. In the final fourth step we then simulate the randomized algorithm on these pivot sets.
\begin{enumerate}
    \setlength\parskip{0pt plus 1pt}
    \item We start constructing $S_0, \dots, S_L$. Initially, pick $S_0 = V$. Then for each level~\smash{$\ell \gets 0, \dots, L-1$} we construct $S_{\ell+1}$ having access to the previously constructed $S_\ell$. We run \cref{lem:min-plus-node-weighted} to compute the distances~\smash{$D^{\leq 2^\ell}[S_\ell, V]$}; additionally, \cref{lem:min-plus-node-weighted} reports for each pair~\makebox{$(s, v) \in S_\ell \times V$} an explicit $s$-$v$-path~$P_{sv}$ of hop-length at most $2^\ell$ realizing the distance~\smash{$D^{\leq 2^\ell}[s, v]$}. Similarly, run \cref{lem:min-plus-node-weighted} to compute $2^\ell$-hop-bounded shortest paths $P_{vs}$ for all pairs $(v, s) \in V \times S_\ell$. Let $\mathcal P$ be the set of such paths $P_{sv}$ and $P_{vs}$ with hop-length exactly $2^\ell$. We now greedily compute a \emph{hitting set} $H$ of~$\mathcal P$, i.e., a set of size $\Order(n \log n / 2^\ell)$ which hits every path in~$\mathcal P$.\footnote{Specifically, greedily include into $H$ a vertex that hits the most paths in $\mathcal P$, and afterwards remove these paths from $\mathcal P$. There is always a vertex hitting at least~\smash{$|\mathcal P| \cdot \frac{2^\ell}{n}$} paths (this is the expected number of paths hit if the vertex was chosen randomly), and hence the size of $\mathcal P$ reduces by a factor of~\smash{$1 - \frac{2^\ell}{n}$} in each step. We thus have that~\smash{$|\mathcal P|^{-1} \leq (1 - \frac{2^\ell}{n})^{|H|} \leq \exp(-|H| \cdot \frac{2^\ell}{n})$}, implying that~\smash{$|H| \leq \frac{n}{2^\ell} \cdot \ln |\mathcal P| = \Order(n \log n / 2^\ell)$}. See~\cite{Lovasz75,Chvatal79} for a detailed analysis of this well-understood greedy algorithm.} We then pick $S_{\ell+1} := H$.
    
    We claim that this construction deterministically satisfies the following property for all $k \geq 1$:
    \begin{equation} \label{lem:apsp-node-weighted-det:eq:bridging}
        \min\Big(\,D^{\leq 2^{\ell+1}}[S_\ell, S_\ell],\, D^{\leq 2^\ell}[S_\ell, S_{\ell+1}] \star D^{\leq k}[S_{\ell+1}, S_{\ell+1}] \star D^{\leq 2^\ell}[S_{\ell+1}, S_\ell]\,\Big) \leq D^{\leq k}[S_\ell, S_\ell].
    \end{equation}
    Indeed, take any pair $u, v \in S_\ell$. Let $P$ be a shortest $u$-$v$-path with hop-length at most $k$ and weight $D^{\leq k}[u, v]$, and let $P$ have minimal hop-length among all such paths. If~\makebox{$|P| \leq 2^{\ell+1}$} then the claim is clearly true, so suppose otherwise. Let $x$ be the $2^\ell$-th vertex in $P$, and let $y$ be the $2^\ell$-th last vertex in~$P$. Consider the path $P'$ obtained from $P$ in which we replace the length-$2^\ell$ prefix by $P_{ux}$ and in which we replace the length-$2^\ell$ suffix by $P_{yv}$. Recall that $P_{ux}$ and $P_{yv}$ have minimal weight among all length-$2^\ell$ paths from $u$ to $x$ and from $y$ to $v$, respectively. Therefore, the weight and hop-length can only have decreased going from $P$ to $P'$. However, we picked $P$ to have smallest weight and minimal hop-length, thus both weight and hop-length remain unchanged. In particular, it follows that $P_{ux}$ and~$P_{yv}$ have length exactly $2^\ell$. Thus $P_{ux}, P_{yv} \in \mathcal P$, and by construction $S_{\ell+1}$ must hit both paths, say in vertices $s$ and $t$. Therefore, finally, we can partition $P'$ into an initial $u$-$s$-segment of hop-length at most $2^\ell$, followed by an $s$-$t$-segment of hop-length at most $|P'| = |P| \leq k$, followed by a trailing $t$-$v$-segment of hop-length at most $2^\ell$. The claim follows.
    \item Next, our goal is to compute, for each pair $u, v \in V$, a path $Q_{uv}$ with hop-length at most $3 \cdot 2^L$ and with weight at most $D^{\leq 2^L}[u, v]$. This step again runs in levels $\ell \gets L, L-1, \dots, 0$, where at the $\ell$-th level the goal is to compute the paths $Q_{uv}$ for all pairs $u, v \in S_\ell$. Eventually, at level $\ell = 0$, we have then computed all paths $Q_{uv}$.
    \begin{itemize}
        \item \emph{Base level $\ell = L$:} Compute $D^{\leq 2^L}[S_L, S_L]$ by \cref{lem:min-plus-node-weighted}. At the same time, \cref{lem:min-plus-node-weighted} yields paths $Q_{vu}$ realizing these distances for all $v, u \in S_L$.
        \item \emph{Levels $\ell \gets L-1, \dots, 0$:} We assume that the paths $Q_{uv}$ for $u, v \in S_{\ell+1}$ have already been computed. On the one hand, (i) by \cref{lem:min-plus-node-weighted} we compute paths realizing the distances~\smash{$D^{\leq 2^{\ell+1}}[S_\ell, S_\ell]$}. On the other hand, (ii) by two applications of \cref{lem:min-plus-node-weighted} (as in \cref{lem:apsp-node-weighted-rand}) we compute paths realizing the distances~\smash{$D^{\leq 2^\ell}[S_\ell, S_{\ell+1}] \star D^{\leq 2^L}[S_{\ell+1}, S_{\ell+1}] \star D^{\leq 2^\ell}[S_{\ell+1}, S_\ell]$}---that is, $u$-$v$-paths for all nodes~\makebox{$u, v \in S_\ell$} that start with an initial $u$-$s$-path of hop-length at most $2^\ell$ leading to some node~\makebox{$s \in S_{\ell+1}$}, followed by a middle segment of hop-length at most $2^L$ to some other node~\makebox{$t \in S_{\ell+1}$}, followed by a trailing $t$-$v$-path of hop-length at most $2^\ell$. Here, we substitute the middle segment from $s$ to $t$ by the precomputed path $Q_{st}$ (of weight at most~\smash{$D^{\leq 2^L}[s, t]$}). For each $u, v \in S_\ell$,~\eqref{lem:apsp-node-weighted-det:eq:bridging} (applied with~\smash{$k = 2^L$}) implies that either the path computed in (i) or in (ii) has weight at most~\smash{$D^{\leq 2^L}[u, v]$}; thus, we choose the path of smaller weight to be $Q_{uv}$.
    \end{itemize}
    Finally, we verify that all paths $Q_{uv}$ indeed have length at most $3 \cdot 2^L$. Indeed, each such path is composed by an initial length-$2^L$ segment (computed in the base case), preceded and succeeded by length-$2^\ell$ segments for each level $\ell < L$. Therefore, the total hop-length of $Q_{uv}$ is $2^L + 2 \cdot \sum_{\ell < L} 2^\ell \leq 3 \cdot 2^L$.
    \item Let $\mathcal Q$ be the set of all the previously computed paths $Q_{uv}$ with hop-length at least $2^L$. We compute a size-$\Order(n \log n / 2^L)$ hitting set $H$ of $\mathcal P$, pick $S^* := S_L \cup H$, and claim that $S^*$ satisfies
    \begin{equation} \label{lem:apsp-node-weighted-det:eq:base}
        D[S^*, S^*] = (D^{\leq 4 \cdot 2^L}[S^*, S^*])^{\star n}
    \end{equation}
    To see this, take any pair $u, v \in S^*$ and take a shortest $u$-$v$-path $P$ of minimal hop-length. Partition~$P$ into consecutive subpaths of length exactly $2^L$ (except perhaps the last subpath which has length at most $2^L$). Then construct another path $P'$ by replacing each such subpath starting at $x$ and ending at $y$ by $Q_{xy}$ (except for the trailing subpath which we leave untouched). On the one hand, it is clear that the weights of $P$ and $P'$ are the same, since we only replaced length-$2^L$ subpaths from $x$ to $y$ by alternative $x$-$y$-paths of weight $D^{\leq 2^L}[x, y]$. On the other hand, we claim that in $P'$ there is no segment of length more than $4 \cdot 2^L$ that avoids $S^*$. Indeed, whenever we replaced a subpath in $P$ by some path $Q_{xy}$, then $Q_{xy}$ must have hop-length at least $2^L$ (otherwise this would contradict the minimality of~$P$). It follows that $Q_{xy} \in \mathcal Q$ for all relevant paths $Q_{xy}$, and thus all these paths are hit by $H \subseteq S^*$. Moreover, each path $Q_{xy}$ has hop-length at most $3 \cdot 2^L$, thus there can be at most $3 \cdot 2^L$ consecutive internal nodes in $P'$ that avoid $S^*$ plus at most $2^L$ nodes in the trailing segment. This completes the proof as~\smash{$w(P') = (D^{\leq 4 \cdot 2^L}[S^*, S^*])^{\star n}[u, v]$}. 
    \item In this final step we simulate our previous randomized algorithm. That is, our goal is to compute the distance matrices $D[S_L, S_L], \dots, D[S_0, S_0]$ from bottom to top:
    \begin{itemize}
        \item \emph{Base level $\ell = L$:} We compute the distances $D^{\leq 4 \cdot 2^L}[S^*, V]$ by one call to \cref{lem:min-plus-node-weighted} with hop-parameter $h = 4 \cdot 2^L$. Then, based on~\eqref{lem:apsp-node-weighted-det:eq:base} we compute the distances~\smash{$D[S^*, S^*] = (D^{\leq 4 \cdot 2^L}[S^*, S^*])^{\star n}$} by repeated squaring, i.e., by $\log n$ min-plus computations of size $|S^*| \times |S^*| \times |S^*|$. As $S_L \subseteq S^*$, we have in particular computed the distances $D[S_L, S_L]$.
        \item \emph{Levels $\ell \gets L-1, \dots, 0$:} Assuming that the distance matrix~\smash{$D[S_{\ell+1}, S_{\ell+1}]$} has already been computed, we compute the distance matrix~\smash{$D[S_\ell, S_\ell]$} as before with three calls to \cref{lem:min-plus-node-weighted} based on the following fact (which is due to~\eqref{lem:apsp-node-weighted-det:eq:bridging} applied with $k = n$):
        \begin{equation*}
            D[S_\ell, S_\ell] = \min\Big(\,D^{\leq 2^{\ell+1}}[S_\ell, S_\ell],\, D^{\leq 2^\ell}[S_\ell, S_{\ell+1}] \star D[S_{\ell+1}, S_{\ell+1}] \star D^{\leq 2^\ell}[S_{\ell+1}, S_\ell]\,\Big).
        \end{equation*}
    \end{itemize}
\end{enumerate}
This completes the description of the algorithm. We have argued the correctness throughout, but it remains to analyze the running time. In the steps~1,~2, and~4 we have applied \cref{lem:min-plus-node-weighted} a constant number of times for each level $0 \leq \ell \leq L$ with hop-parameter $h = \Order(2^\ell)$. All in all, this takes time~\smash{$\widetilde\Order(\sum_{\ell=0}^L 2^\ell \cdot T_{\ostar}(n / 2^\ell, n, n))$}. In step 4 we additionally compute a logarithmic number of min-plus products in time~\smash{$\widetilde\Order(|S^*|^3)$}; since by construction we have~\smash{$|S^*| \leq \widetilde\Order(n / 2^L) = \widetilde\Order(n / h)$} this becomes~\smash{$\widetilde\Order(n^3 / h^3)$}. Step 3 and all other computations run in negligible time $\widetilde\Order(n^2 \cdot 2^L)$. The claimed time bound therefore follows by the same analysis as in the proof of \cref{lem:apsp-node-weighted-rand}.
\end{proof}

As before, we combine \cref{lem:apsp-node-weighted-det} with \cref{lem:min-plus-boolean} to obtain a deterministic algorithm for Node-Weighted APSP. This completes the proof of \cref{thm:apsp-node-weighted}.

\subsection{Negative or Real Weights} \label{sec:apsp-node-weighted:sec:negative}
So far we have focused on the nonnegative-weight case, but it is easy to verify that the algorithm works equally well in the presence of negative weights, assuming that the graph contains no negative cycles. In the presence of negative cycles, one can either modify the algorithm to detect these (in the repeated squaring step we have to watch out for entries which keep decreasing), or alternatively eliminate all negative cycles in a preprocessing step.\footnote{Specifically, let $G$ be the given graph with possibly negative weights. Compute the strongly connected components $C_1, \dots, C_k$ in $G$, and decide for each such component whether it contains a negative cycle (say by~\cite{BernsteinNW22} in time~\smash{$\widetilde\Order(\sum_i |C_i|^2) = \widetilde\Order(n^2)$}). Then replace each such component $C_i$ by a single node of highly negative weight---say, $-2W n$ where $W$ is the largest weight in $G$ in absolute value---and compute Node-Weighted APSP in the resulting graph $G'$. This eliminates all negative cycles, and we can read off the distances in $G$ from the distances in $G'$. More precisely, for each pair~$u, v$ with $D_{G'}[u, v] < -W n$, the shortest path must have crossed a component with a negative cycle and thus $D_G[u, v] = -\infty$; for all other pairs we clearly have~\makebox{$D_{G'}[u, v] = D_G[u, v]$}.} Moreover, our algorithm works for real weights (in the real RAM model) without any further modifications.

\subsection{Extension to Few-Weights APSP} \label{sec:apsp-node-weighted:sec:d-weights}
Finally, we will extend this algorithm from Node-Weighted (aka $1$-Weights APSP) to $d$-Weights APSP for any parameter $d \geq 1$. The only modification is that we need to replace our algorithm for Boolean (aka $1$-Weights) Min-Plus product by an algorithm for $d$-Weights Min-Plus product. The following lemma generalizes the algorithm for rectangular Boolean Min-Plus Product to rectangular $d$-Weights Min-Plus Product. Again, the proof is a simple adaption of~\cite[Lemma~3.3]{Yuster09}.

\begin{lemma}[Rectangular $d$-Weights Min-Plus Product] \label{lem:d-weights-min-plus}
Let $A \in (\Int \cup \set{\infty})^{s \times n}$ and let $B \in (\Int \cup \set{\infty})^{n \times n}$ such that each column in $B$ contains at most $d$ distinct non-$\infty$ entries. There is a deterministic algorithm computing $A \star B$ in time
\begin{equation*}
    T_d(s, n, n) = \widetilde\Order\parens*{\min_{\Delta \geq 1} \parens*{\MM(s\Delta, n, n d) + \frac{s n^2 d}{\Delta}}}.
\end{equation*}
\end{lemma}
\begin{proof}
For every $i \in [s]$, we again sort the entries of the $i$-th row of $A$ to form a list $L_i$, and split $L_i$ (in sorted order) into $\Delta$ buckets of $\Theta(n/\Delta)$ entries. Denoting by $L_{i, b}$ be the entries in the $b$-th bucket, we form an $(s\Delta) \times n$ Boolean matrix $A'$ defined as
\begin{equation*}
    A'[(i, b), k] =
    \begin{cases}
        1 &\text{if $A[i, k] \in L_{i, b}$,} \\
        0 &\text{otherwise.}
    \end{cases}
\end{equation*}
Next, let $w_{k, \ell}$ denote the $\ell$-th distinct entry in the $k$-th row of $B$ (arbitrarily ordered). We form an $n \times (nd)$ Boolean matrix $B'$ defined by
\begin{equation*}
    B'[k, (j, \ell)] =
    \begin{cases}
        1 &\text{if $B[k, j] = w_{k, \ell}$}, \\
        0 &\text{otherwise,}
    \end{cases}
\end{equation*}
and multiply $A'$ by $B'$ (over the Boolean semiring) in time $\MM(s\Delta, n, nd)$. For every pair~\smash{$i \in [s], j \in [n]$} we can now determine $(A \star B)[i, j]$ as follows: For all $\ell \in [d]$, find the smallest bucket $b$ such that $(A' B')[(i, b), (j, \ell)] = 1$ and then search through the list $L_{i, b}$ for the smallest entry $A[i, k] \in L_{i, b}$ such that $B[k, j] = w_{k, \ell}$. Remember $A[i, k] + B[k, j]$ as the \emph{candidate} for $\ell$; then $(A \star B)[i, j]$ equals the smallest candidate across all iterations $\ell$. This takes time $O(snd \cdot n/\Delta)$ overall. The correctness should be clear, and the running time is $\Order(\MM(s \Delta, n, nd) + sn^2 d / \Delta)$ (where the parameter $\Delta \geq 1$ can be freely chosen).
\end{proof}

As a consequence we obtain the analogue of \cref{lem:min-plus-node-weighted} in the $d$-weights setting (by essentially the same proof, i.e., repeated computation of rectangular $d$-Weights Min-Plus Products):

\begin{lemma} \label{lem:min-plus-d-edge-weights}
Let $G = (V, E, w)$ be a directed edge-weighted graph in which each node has incoming edges of at most $d$ distinct weights. Let $S \subseteq V$, $A \in (\Int \cup \set{\infty})^{S \times V}$, and $h \geq 1$. Then we can compute the matrix~\smash{$A \star D^{\leq h}_G$} in deterministic time~\smash{$\widetilde\Order(h \cdot T_d(|S|, n, n))$}. In the same time we can compute for each pair~\makebox{$(s, v) \in S \times V$} a \emph{witness} path $v_0, \dots, v_\ell$ in $G$ of hop-length $\ell \leq h$ minimizing~\makebox{$A[s, v_0] + w(v_0, v_1) + \dots + w(v_{\ell-1}, v_\ell)$}.
\end{lemma}

To generalize our algorithm from Node-Weighted APSP to $d$-Weights APSP, it suffices to consistently replace all calls to \cref{lem:min-plus-node-weighted} by a call to \cref{lem:min-plus-d-edge-weights}. We thereby obtain:

\begin{lemma} \label{lem:d-weights-apsp-det}
$d$-Weights APSP is in deterministic time
\begin{equation*}
    \widetilde\Order\parens*{\min_{h \geq 1} \parens*{\frac{n^3}{h^3} + h \cdot T_d(n / h, n, n)}}.
\end{equation*}
In other words, the algorithm is obtained via an~\smash{$\widetilde\Order(n^3 / h^3)$}-time reduction from $d$-Weights APSP to~\smash{$\widetilde\Order(h)$} calls to an oracle for $(n / h) \times n \times n$ $d$-Weights Min-Plus products.
\end{lemma}

In particular, by combining \cref{lem:d-weights-min-plus,lem:d-weights-apsp-det} we obtain a deterministic algorithm for $d$-Weights APSP running in time
\begin{align*}
    \widetilde\Order\parens*{\frac{n^3}{h^3} + h \cdot T_d(n / h, n, n)}
    &= \widetilde\Order\parens*{\frac{n^3}{h^3} + h \cdot \parens*{\MM((n/h)\Delta, n, nd) + \frac{(n / h) n^2 d}{\Delta}}} \\
    &= \widetilde\Order\parens*{\frac{n^3}{h^3} + h \cdot \MM(n\Delta/h, n, nd) + \frac{n^3 d}{\Delta}},
\end{align*}
where the parameters $h \geq 1$ and $\Delta \geq 1$ can be chosen arbitrarily. For the crude upper bound~\smash{$\widetilde\Order(d \cdot n^{(3+\omega)/2})$}, it suffices to pick $h = \Delta = n^{(3-\omega)/2}$ (using that $\MM(n, n, nd) \leq \Order(d \cdot n^\omega)$). The optimal trade-off, assuming that $d = \Theta(n^\delta)$ for some $0 \leq \delta \leq 1$, is to choose $h = n^\gamma$ and $\Delta = n^{3\gamma+\delta}$ where $\gamma$ is the solution to the equation $\gamma + \omega(1+\delta+2\gamma, 1, 1 + \delta) = 3 - 3\gamma$. This completes the proof of \cref{thm:apsp-few-weights-naive}.
\section{APSP with Few Distinct Weights} \label{sec:apsp-few-weights}
Our goal in this section is to prove \cref{thm:apsp-few-weights}, i.e., to show that for any $\delta > 0$, the $n^{(3-\omega)-\delta}$-Weights APSP problem can be solved in subcubic time $\Order(n^{3-\epsilon})$ for some $\epsilon > 0$. The proof is structured as outlined in the overview: We first reduce from APSP to the All-Edges Exact Triangle problem (\cref{sec:apsp-few-weights:sec:exact-tri}), then we design algorithms for the ``small-doubling'' (\cref{sec:apsp-few-weights:sec:small-doubling}) and uniform-regular Exact Triangle (\cref{sec:apsp-few-weights:sec:uniform-regular}), and then establish the uniformization (\cref{sec:apsp-few-weights:sec:uniformization}) and regularization (\cref{sec:apsp-few-weights:sec:regularization}) reductions.

\subsection{From APSP to Exact Triangle} \label{sec:apsp-few-weights:sec:exact-tri}

\begin{definition}[All-Edges Exact Triangle] \label{def:exact-tri}
Let $A, B, C \in (\Int \cup \set{\bot})^{n \times n}$. We say that a triple $(i, k, j)$ is an \emph{exact triangle} if $A[i, k] + B[k, j] = C[i, j]$, and we denote by $\mathcal T(A, B, C) \subseteq [n]^3$ the set of all exact triangles. The \emph{All-Edges Exact Triangle} problem is, given $(A, B, C)$, to decide for each pair $(i, j) \in [n]^2$ if there is some exact triangle $(i, k, j) \in \mathcal T(A, B, C)$.
\end{definition}

Recall that we consider the symbol ``$\bot$'' as a place-holder symbol for a missing edge (so that, e.g., any pair~$(i, k)$ with $A[i, k] = \bot$ cannot be involved in an exact triangle). As before, we assume that the matrix entries in Exact Triangles are polynomially bounded (i.e., that there is some constant $c$ such that all entries stem from $\set{-n^c, \dots, n^c}$). 

\begin{definition}[$d$-Weights Exact Triangle] \label{def:d-weights-exact-tri}
Let $d \geq 1$. We say that an Exact Triangle instance $(A, B, C)$ is \emph{$d$-weights} if in at least one of the matrices $A, A^T, B, B^T, C, C^T$ all rows contain at most $d$ distinct (non-$\bot$) entries.
\end{definition}

The upcoming reductions are more coarse-grained than our algorithm in the previous section, as in this section we do not mind the precise subcubic running time. For the first reduction from $d$-Weights APSP to $d$-Weights Min-Plus Product we could also take Yuster's algorithm~\cite{Yuster09}, but we include a quick (and slightly more efficient) reduction based on our previous algorithm.

\begin{lemma}[Few-Weights APSP to Min-Plus Product] \label{lem:apsp-to-min-plus}
If $d$-Weights Min-Plus Product can be solved in time $\Order(n^{3-\epsilon})$ (for some $\epsilon > 0$), then $d$-Weights APSP can be solved in time $\Order(n^{3-3\epsilon/4})$.
\end{lemma}
\begin{proof}
Suppose that $d$-Weights Min-Plus Product can be solved in time $\Order(n^{3-\epsilon})$; i.e., in the notation of \cref{sec:apsp-node-weighted:sec:d-weights} we have that $T_d(n, n, n) = \Order(n^{3-\epsilon})$. Our deterministic algorithm for $d$-Weights APSP (\cref{lem:d-weights-apsp-det}) runs in time~\smash{$\widetilde\Order(n^3 / h^3 + h \cdot T_d(n/h, n, n)) = \widetilde\Order(n^3 / h^3 + h \cdot n^{3-\epsilon})$}, so the claim follows by picking $h = n^{\epsilon/4}$.
\end{proof}

The next reduction from $d$-Weights Min-Plus Product to $d$-Weights All-Edges Exact Triangle is a simple adaption of the reduction in~\cite{WilliamsW18}. We include a short proof for the sake of a self-contained presentation.

\begin{lemma}[Few-Weights Min-Plus Product to All-Edges Exact Triangle] \label{lem:min-plus-to-exact-tri}
If $d$-weights All-Edges Exact Triangle can be solved in time $T(n, d)$, then $d$-weights Min-Plus Product can be solved in time $\Order(T(n, d) \log n)$.
\end{lemma}
\begin{proof}
Consider matrices $A, B \in (\Nat \cup \set{\bot})^{n \times n}$ (the assumption that the entries are nonnegative is without loss of generality by shifting all entries by a common large number). We compute their min-plus product~$C$ by a recursive algorithm. Let $A'$ and $B'$ be the matrices defined by~\smash{$A'[i, k] = \floor{\frac{A[i, k]}{2}}$} and~\smash{$B' = \floor{\frac{B[k, j]}{2}}$}, and recursively compute the min-plus product $C'$ of $A'$ and $B'$ (note that the Min-Plus Product instance $(A', B')$ is also $d$-weights). From the trivial inequalities $x - 1 \leq \floor{x} \leq x$, it follows that
\begin{align*}
    C[i, j] - 4
    &= \min_{k \in [n]} \parens*{2 \cdot \tfrac{A[i, k]}{2} - 2 + 2 \cdot \tfrac{B[k, j]}{2} - 2} 
    \leq \min_{k \in [n]} \parens*{2 \cdot \floor{\tfrac{A[i, k]}{2}} + 2 \cdot \floor{\tfrac{B[k, j]}{2}}}
    = 2 C'[i, j] \leq C[i, j],
\end{align*}
that is, $2C'$ is an additive approximation of $C$ with error at most $4$ in each entry. Having access to $C'$, we can now exactly compute $C$ by $\Order(1)$ calls to $d$-Weights All-Edges Exact Triangle. Specifically, consider the All-Edges Exact Triangle instances $(A, B, C_\ell)$ for $\ell \in [4]$, where $C_\ell$ be the matrix defined by $C_\ell[i, j] = C'[i, j] + \ell$. It is easy to verify that $C[i, j] = C'[i, j] + \ell$, where $\ell$ is the smallest choice for which $(i, j)$ is involved in an exact triangle in $(A, B, C_\ell)$.

With each recursive call we halve the maximum entry size, thus the recursion depth is at most $\ceil{\log N}$, where $N$ is an upper bound on the initial entry size. Recall that we assume $N = \poly(n)$, and so the time bound $\Order(T(n, d) \log n)$ follows.
\end{proof}

Our goal for the remainder of this section is to prove the following \cref{thm:exact-tri-few-weights}, which, combined with \cref{lem:apsp-to-min-plus,lem:min-plus-to-exact-tri}, completes the proof of \cref{thm:apsp-few-weights}.

\thmexacttrifewweights*

\subsection{Small-Doubling Exact Triangle} \label{sec:apsp-few-weights:sec:small-doubling}
We start with the following lemma for the ``small-doubling'' case:

\begin{lemma}[Small-Doubling Exact Triangle] \label{lem:exact-tri-small-doubling}
There is a deterministic algorithm that solves All-Edges Exact Triangle on instances $(A, B, C)$ with $A \in (X \cup \set{\bot})^{n \times n}$ and $B \in (Y \cup \set{\bot})^{n \times n}$ in time~\smash{$\widetilde\Order(n^\omega \cdot |X + Y|)$}.
\end{lemma}

This algorithm can be seen as a generalization of the $\widetilde\Order(M \cdot n^\omega)$-time algorithm for APSP (and Exact Triangle) with weights in the range $\set{-M, \dots, M}$~\cite{AlonGM97}, which is based on the following combination of fast matrix multiplication and the fast Fourier transform:

\begin{lemma}[Fast Matrix Multiplication + Fast Fourier Transform~\cite{AlonGM97}] \label{lem:fastmm-fft}
Let $A, B$ be $n \times n$ matrices where each entry is a degree-$d$ polynomial. There is a deterministic algorithm computing the product $A B$ in time~\smash{$\widetilde\Order(n^\omega d)$}.
\end{lemma}

Our strategy is to reduce to this weight-bounded case by linearly \emph{hashing} all weights to $[d]$; this is a common idea in the context of computing sparse convolutions~\cite{ColeH02,ArnoldR15,ChanL15,Nakos20,GiorgiGC20,BringmannFN21,BringmannFN22,JinX24}.\footnote{In fact, there is an alternative \emph{algebraic} approach to deterministically computing sparse convolutions~\cite{BringmannFN22}. This approach works equally well in our context, but we have decided to stick to the hashing-based approach for the sake of simplicity.}

\begin{lemma}[Deterministic Almost-Linear Almost-Perfect Hashing] \label{lem:det-linear-hashing}
Let $Z \subseteq [-N, N]$ have size $t$. There is a deterministic \smash{$\widetilde\Order(t^2 \log N)$}-time algorithm computing a set $P \subseteq [8 t \log N]$ of at most $\ceil{\log t}$ primes satisfying the following property: For each $z \in Z$ there is some prime $p \in P$ such that $z \not\equiv z' \mod p$ for all $z' \in Z \setminus \set{z}$ (in this case we say that $z$ is \emph{isolated modulo $p$}).
\end{lemma}

We remark that this lemma is trivial if we allow randomization: Any prime $p \in [4t \log N, 8t \log N]$ (say) isolates any fixed element $z \in Z$ with constant probability (see the following proof), and thus by taking $\Order(\log t)$ independent samples each element would be isolated at least once.

\begin{proof}
Our proof is a routine application of the method of conditional expectations. Initially we let $P = \emptyset$ and $Z_0 = Z$; intuitively, $Z_0$ is the set of elements for which we still have to find an isolating prime. The algorithm runs in several iterations, stopping only when $Z_0 \neq \emptyset$. In each iteration our goal is to select a prime~$p$ that isolates at least~\smash{$\frac{|Z_0|}{2}$} primes in $Z_0$. To this end, we enumerate all primes in~\makebox{$[m, 2m]$} (where we will choose $m = 4 t \log N$) and pick any prime satisfying this property. Afterwards, we remove all elements from $Z_0$ that are isolated modulo $p$.

For the correctness we prove that there always is a prime $p$ as required. Imagine that $p \in [m, 2m]$ is a uniformly \emph{random} prime. Fix a pair $(z, z') \in Z_0 \times Z$. The event that $z \equiv z' \mod p$ is equivalent to $p$ being a divisor of $z - z'$. On the one hand, $z - z'$ has at most $\log_m(N)$ prime divisors in~$U$. On the other hand, by a quantitative Prime Number Theorem~\cite{RosserS62} there are at least $\frac{3m}{5 \ln m}$ primes in~$[m, 2m]$ (for sufficiently large~$m$). Combining both statements, the probability that $z \equiv z' \mod p$ is at most~\smash{$\log_m(N) / \frac{3m}{5 \ln m} \leq \frac{5 \log N}{3m} \leq \frac{1}{2t}$} by choosing $m := 4 t \log N$. In particular, the expected number of pairs $(z, z') \in Z_0 \times Z$ with $z \equiv z' \mod p$ is at most~\smash{$\frac{1}{2t} \cdot |Z_0| \cdot |Z| = \frac{|Z_0|}{2}$}. Hence, there is some choice of $p$ which isolates at least~\smash{$\frac{|Z_0|}{2}$} many elements, as claimed.

For the running time, observe that the algorithm terminates after at most $\ceil{\log t}$ iterations (and thus also $|P| \leq \ceil{\log t}$). In each iteration we enumerate $\Order(t \log N)$ primes $p$, and for each such prime compute the set of isolated elements in time~\smash{$\widetilde\Order(t)$} (say, by sorting the set $\set{z \bmod p : z \in Z}$).
\end{proof}

\begin{proof}[Proof of \cref{lem:exact-tri-small-doubling}]
We can assume that $|X|, |Y| \leq |X + Y| \leq n$, as otherwise we can simply solve the given instance in time $\Order(n^3) = \Order(n^\omega \cdot |X + Y|)$. As a first step we compute the sumset $X + Y$ (by brute-force in time~$\Order(n^2)$). Let $t := |X + Y|$ and let $N$ denote the maximum entry in~$X + Y$ (in absolute value). We then apply \cref{lem:det-linear-hashing} with $Z := X + Y$ to compute a set of primes $P$. For each $p \in P$, we execute the following steps. Let $x$ be a formal variable, and let $A_p, B_p$ be matrices defined by
\begin{align*}
    A_p[i, k] &= x^{A[i, k] \bmod p}, \\
    B_p[k, j] &= x^{B[k, j] \bmod p}.
\end{align*}
That is, $A_p$ and $B_p$ are matrices where each entry is a univariate degree-$p$ polynomial. We can compute their product $A_p B_p$ in time \smash{$\widetilde\Order(n^\omega p)$} by \cref{lem:fastmm-fft}, and by reducing the exponents in the resulting matrix modulo~$p$ we have access to the matrix $C_p$ with
\begin{equation*}
    C_p[i, j] = \sum_{k \in [n]} x^{(A[i, k] + B[k, j]) \bmod p}.
\end{equation*}
We can now decide for each pair $(i, j)$ whether it appears in an exact triangle $(i, k, j)$ as follows: It is clear that if $C[i, j] \not\in X + Y$, then the answer is trivially ``no''. Otherwise, if $C[i, j] \in X + Y$, by \cref{lem:det-linear-hashing} there is some prime $p$ which isolates $C[i, j]$ from the rest of $X + Y$. We return ``yes'' if and only if the coefficient of~$x^{C[i, j] \bmod p}$ in $C_p[i, j]$ is nonzero. To see that this is correct, note that the coefficient of the $x^c$-monomial in $C_p[i, j]$ is nonzero if and only if there is a triple $(i, k, j)$ with $A[i, k] + B[k, j] \equiv c \mod{p}$. If $c$ is isolated modulo $p$ in $X + Y$, this is further equivalent to $A[i, k] + B[k, j] = c$.
\end{proof}

We remark that this algorithm only decides for each pair $(i, j)$ whether it is involved in an exact triangle, and this is sufficient for us. However, by incorporating the ``Baur-Strassen'' trick by~\cite{Fischer24,BaurS83} one could additionally decide for each pair $(i, k)$ and $(k, j)$ if it is involved in an exact triangle, in the same running time.

\subsection{Uniform Regular Exact Triangle} \label{sec:apsp-few-weights:sec:uniform-regular}
As a first step we design a subcubic algorithm for the significantly more structured case that the given Exact Triangle instance is not only $d$-weights, but even satisfies the following two properties:

\begin{definition}[$d$-Uniform Exact Triangle] \label{def:d-uniform}
Let $d \geq 1$. We say that a Exact Triangle instance $(A, B, C)$ is \emph{$d$-uniform} if all entries in $A$, $B$ and $C$ are from $W \cup \set{\bot}$, where $W \subseteq \Int$ is a set of size at most $d$.
\end{definition}

\begin{definition}[$r$-Regular Exact Triangle] \label{def:r-regular}
Let $r \geq 1$. We say that a matrix $A \in (\Int \cup \set{\bot})^{n \times n}$ is  \emph{$r$-regular} if each non-$\bot$ entry appears at most $r$ times in its respective row \emph{and} column. We say that an Exact Triangle instance $(A, B, C)$ is \emph{$r$-regular} if all three matrices $A$, $B$ and $C$ are $r$-regular.
\end{definition}

That is, the goal of this subsection is to prove the following lemma:

\lemexacttriuniformregular*

Note that if $\omega=2$, the running time is~\smash{$\widetilde\Order(n^{3-1/7} d^{1/7})$}. For the proof we use Chan and Lewenstein's BSG Covering; this follows from~\cite[Corollary 2.2, Theorem 2.3]{ChanL15}.

\thmbsgcovering*

\begin{proof}[Proof of \cref{lem:exact-tri-uniform-regular}]
Let $X$, $Y$, $Z$ denote the set of entries in $A$, $B$ and $C$, respectively. By the $d$-uniformity assumption we have that $|X|, |Y|, |Z| \leq d$. We apply the BSG Covering Theorem to~$X, Y, Z$ with some parameter $K \geq 1$. This yields sets $X_1, \dots, X_K \subseteq X$ and $Y_1, \dots, Y_K \subseteq Y$ and a set of pairs $R \subseteq X \times Y$. We then run the following two steps:

\begin{enumerate}
    \item \emph{(The additively structured case)} For each $k \in [K]$, let $A_k$ denote the matrix obtained from $A$ by restricting to the entries in $X_k$ (and replacing all other entries by $\bot$), and similarly let $B_k$ denote the matrix obtained from $B$ by restricting to the entries in $Y_k$. We solve the Exact Triangle instance $(A_k, B_k, C)$ by \cref{lem:exact-tri-small-doubling}, and report all ``yes''-pairs $(i, j)$.
    \item \emph{(The remainder case)} Next, we consider the pairs in $R$. To this end we simply enumerate all $(a, b) \in R$ and $k \in [n]$, and then enumerate all indices $i$ with $A[i, k] = a$ and all indices $j$ with $B[k, j] = b$. We test for each such triple $(i, k, j)$ whether it forms an exact triangle and report accordingly.
\end{enumerate}

For the correctness, first note that the algorithm clearly reports ``no'' for all pairs $(i, j)$ not involved in exact triangles. We argue that for each pair $(i, j)$ involved in an exact triangle, the algorithm reports ``yes''. Fix such an exact triangle $(i, k, j)$. If $(A[i, k], B[k, j]) \in R$, then we necessarily list the triangle in step~2. Otherwise, given that $A[i, k] \in X$, $B[k, j] \in Y$ and $C[i, j] = A[i, k] + B[k, j] \in Z$, Property~(i) of the BSG Covering implies that there is some index $k' \in [K]$ such that $(A[i, k], B[k, j]) \in X_{k'} \times Y_{k'}$. We therefore report the pair $(i, j)$ in step~1.

Finally, we analyze the running time. Step~1 runs for $K$ iterations, each running in time~\smash{$\widetilde\Order(n^\omega \cdot |X_k + Y_k|)$} by \cref{lem:exact-tri-small-doubling}. Property~(ii) of the BSG Covering guarantees that $|X_k + Y_k| \leq \Order(d K^5)$, leading to total time~\smash{$\widetilde\Order(n^\omega \cdot d K^6)$}. For step~2, we use that each column in $A$ contains each entry $a$ at most $n/d$ times (by the $d$-regularity assumption), and similarly that each row in $B$ contains each entry $b$ at most $n/d$ times. The running time of step~2 is thus bounded by $\Order(|R| \cdot n \cdot (n/d)^2) = \Order(n^3 / K)$, using that $|R| \leq \Order(d^2 / K)$ by Property~(iii) of the BSG Covering. In total running time is~\smash{$\widetilde\Order(n^\omega K^6 d + n^3 / K)$}. The claim follows by setting $K = (n^{3-\omega} / d)^{1/7}$ as then $n^3/K = n^{3-(3-\omega)/7}d^{1/7}$.
\end{proof}

In the following two subsections we will separately justify the uniformity and regularity assumptions.

\subsection{Uniformization} \label{sec:apsp-few-weights:sec:uniformization}
\begin{definition}[Popular Sums]
Let $X, Y \subseteq \Int$. We write $r_{X+Y}(z) = \abs{\set{(x, y) \in X \times Y : x + y = z}}$. For a threshold $t \geq 1$, the set of \emph{popular sums} is defined as
\begin{equation*}
    P_t(X, Y) = \set{z \in X + Y : r_{X+Y}(z) \geq t}.
\end{equation*}
\end{definition}

\begin{lemma}[Approximating Popular Sums] \label{lem:approx-popular-sums}
Let $X, Y \subseteq \Int$ be sets of size at most $d$. For any $t \geq 1$ we can compute a set $P_{2t}(X, Y) \subseteq P \subseteq P_t(X, Y)$ by a randomized algorithm in time $\widetilde\Order(d^2 / t)$, or by a deterministic algorithm in time~\smash{$\widetilde\Order(d^2/t)\cdot U^{o(1)}$}, where $U= \max_{u\in X\cup Y}|u|$.
\end{lemma}
\begin{proof}
Uniformly subsample the set $X$ with rate $p = \Theta(\log d/\sqrt{t})$ to obtain a subset $X'$, and subsample~$Y$ with the same rate $p$ to obtain $Y'$. For each sum $z \in X + Y$, the expected number of representations in $X' + Y'$ is exactly $p^2 \cdot r_{X+Y}(z)$. By Chernoff bound, any sum $z$ with at least $2t$ representations satisfies that $r_{X'+Y'}(z) \geq 1.9p^2 t$ with high probability, whereas any sum $z$ with at most $t$ representations satisfies that $r_{X'+Y'}(z) \leq 1.1p^2 t$. Thus we take $P = \set{z \in X' + Y' : r_{X' + Y'}(z) \geq 1.5 p^2 t}$. To compute $P$, note that we have that~\smash{$|X'|, |Y'| \leq \widetilde\Order(d / \sqrt{t})$}, and thus we can compute $X' + Y'$ (along with all multiplicities) by brute-force in time~\smash{$\widetilde\Order(d^2 / t)$}. 

To get a deterministic algorithm, we use a recent result from \cite[Section 3]{FischerJX25}, which states that given integer sets $X,Y \subset [-U,U]$ and parameter $\eps>0$, one can deterministically compute a vector $f$ with at most $O(\eps^{-1}|X|\log ^2 U)$ nonzero entries in $(\eps^{-1}|X|+|Y|)\cdot U^{o(1)}$ time, such that for all $z\in \Int$, $|f[z]-r_{X+Y}(z)| \le \eps |Y|$. We apply this result setting $\eps = \frac{t}{2d}$, and return $P= \{z: f[z]\le 1.5t\}$. The time complexity is $O(d^2/t)\cdot U^{o(1)}$. For all $z$ with $r_{X+Y}(z)\ge 2t$ we have $f[z]\ge r_{X+Y}(z)-\eps |Y| \ge 2t- 0.5t=1.5t$, so $P_{2t}(X,Y) \subseteq P$. For all $z$ with $r_{X+Y}(z)<t$ we have $f[z]\le r_{X+Y}(z)+\eps|Y| < t+0.5t=1.5t$, so $P \subseteq P_t(X,Y)$.
\end{proof}

\lempopsumsdecomp*
\begin{proof}
Note that the statement is symmetric with respect to the partition of the sets $X_i$ and $Y_j$. In the following we thus only show that the decomposition of the sets $X_i$ exists as claimed (satisfying the first statements in items~(1) and (2), respectively). Then from the same proof applied to the instance where we exchange the $X_i$'s and $Y_j'$s, the full statement follows.

Consider the following algorithm that proceeds in \emph{iterations} $\ell := 1, 2, \dots, \Delta^2$. In the $\ell$-th iteration our goal is to assign the parts $X_{i, \ell}$ (and the accompanying sets $S_\ell$ and translates $s_{i, \ell}$). To this end, our first step is to apply \cref{lem:approx-popular-sums} (with parameter $t := d / \Delta$) for each pair $(i, j) \in [n]^2$ to compute sets $P_{i, j}$ satisfying that
\begin{equation*}
    P_{2d / \Delta}(X_i, Y_j) \subseteq P_{i, j} \subseteq P_{d / \Delta}(X_i, Y_j).
\end{equation*}
We say that an index $j$ has \emph{degree} $\deg(j) = \abs{\set{i \in [n] : P_{i, j} \neq \emptyset}}$. We then distinguish two cases:
\begin{itemize}
    \item \emph{There is an index $j$ with degree at least $n / \Delta$:} We pick $S_\ell := -B_j$. For each $i$ with $P_{i, j} \neq \emptyset$ we choose an arbitrary popular sum $s_{i, \ell} \in P_{i, j}$ and assign $X_{i, \ell} := X_i \cap (s_{i, \ell} + S_\ell)$. For all other indices $i$ (with~$P_{i, j} = \emptyset$) we pick $X_{i, \ell} := \emptyset$ and an arbitrary translate $s_{i, \ell}$. Then we update $A_i := A_i \setminus A_{i, \ell}$ and continue with the next iteration.
    \item \emph{All indices $j$ have degree less than $n / \Delta$:} We stop the process and we assign the final piece of the decomposition as $X_i' := X_i$. (And for completeness we pick $X_{i, \ell} := \emptyset$ for all pieces $X_{i, \ell}$ that have not been assigned yet.)
\end{itemize}

We analyze the correctness of this algorithm. Property (1) is obvious by construction. Property (2) is also obvious: When we assign $X_i' = X_i$ we have reached a point where all indices $j$ have degree at most~$n / \Delta$, and thus
\begin{equation*}
    \abs{\set{(i, j) \in [n] : P_{2d / \Delta}(X_i', Y_j) \neq \emptyset}} \leq \abs{\set{(i, j) \in [n] : P_{i, j} \neq \emptyset}} = \sum_{j \in [n]} \deg(j) \leq \frac{n^2}{\Delta},
\end{equation*}
given that $P_{2d / \Delta}(X_i', Y_j) = P_{2d / \Delta}(X_i, Y_j) \subseteq P_{i, j}$. It remains to show that the algorithm indeed terminates after at most $\Delta^2$ iterations. We consider the total size~$\sum_i |X_i|$ as our progress measure and argue that with each iteration this measure decreases by a significant amount. Specifically, focus on an arbitrary iteration~$\ell$ and let $j$ be the chosen high-degree index. Then there are at least $n / \Delta$ indices $i$ with $P_{i, j} \neq \emptyset$. For each such index we choose a set $X_{i, \ell}$ of size at least $d / \Delta$ given that $s_{i, \ell} \in P_{i, j} \subseteq P_{d / \Delta}(X_i, Y_j)$ is a $d / \Delta$-popular sum and thus $|X_{i, \ell}| = |X_i \cap (s_{i, \ell} - B_j)| \geq d / \Delta$. Therefore, in each iteration our measure $\sum_i |A_i|$ decreases by at least $(n / \Delta) \cdot (d / \Delta)$. Since initially $\sum_i |A_i| \leq nd$, there can be at most $\Delta^2$ iterations.

Finally, we analyze the running time. In each iteration $\ell$ we first call \cref{lem:approx-popular-sums} $n^2$ times each with parameter $t = d / \Delta$. By a randomized algorithm this takes time $\widetilde\Order(n^2 \cdot d^2 / t) = \widetilde\Order(n^2 d \Delta)$, and by a deterministic algorithm the overhead is $U^{\order(1)}$. Afterwards, the degrees can be computed in time $\Order(n^2)$, and assigning the shifts~$s_{i, \ell}$ and sets~$X_{i, \ell}$ takes time $\Order(n d)$. Summing over all $\Delta^2$ iterations, the claimed time bound follows.
\end{proof}

\begin{observation}[Naive Uniformization] \label{obs:exact-tri-uniformization-naive}
Let $d, \Delta \geq 1$, and let $(A, B, C)$ be a $d \Delta$-uniform Exact Triangle instance. We can compute $d$-uniform Exact Triangle instances~\smash{$\set{(A_\ell, B_\ell, C_\ell)}_{\ell=1}^{\Delta^3}$} such that $\mathcal T(A, B, C) = \bigsqcup_{\ell \in [\Delta^3]} \mathcal T(A_\ell, B_\ell, C_\ell)$ in time $\Order(n^2 \Delta^3)$.
\end{observation}
\begin{proof}
Arbitrarily partition the set of entries in $A$ into $\Delta$ parts of size at most $d$, and let $A_1, \dots, A_\Delta$ be the matrices obtained from $A$ by restricting to these entries respectively (replacing other entries by~$\bot$). Similarly, partition $B$ into $B_1, \dots, B_\Delta$ and $C$ into $C_1, \dots, C_\Delta$. Then consider the $d$-uniform instances~$(A_x, B_y, C_z)$ for~$x, y, z \in [\Delta]$; clearly each exact triangle in the original instance appears in exactly one such instance.
\end{proof}

\begin{lemma}[Uniformization] \label{lem:exact-tri-uniformization}
Let $d, \Delta \geq 1$, and let $(A, B, C)$ be a $d$-weights Exact Triangle instance. Then we can compute Exact Triangle instances \smash{$\set{(A_\ell, B_\ell, C_\ell)}_{\ell=1}^{L}$} and a set of triples $T \subseteq [n]^3$ such that:
\begin{enumerate}[label=(\roman*)]
    \item $\mathcal T(A, B, C) = T \sqcup \bigsqcup_{\ell \in [L]} \mathcal T(A_\ell, B_\ell, C_\ell)$.
    \item For each $\ell \in [L]$, the instance $(A_\ell, B_\ell, C_\ell)$ is $d$-uniform.
    \item $L \leq \poly(\Delta \log n)$ and~\smash{$|T| \leq \widetilde\Order(n^3 / \Delta)$}.
\end{enumerate}
The algorithm runs in time $\widetilde\Order(n^{2+\order(1)} d \poly(\Delta) + n^3 / \Delta)$.
\end{lemma}
\begin{proof}
The $d$-weights assumption implies that in at least one of the matrices $A, A^T, B, B^T, C, C^T$, each row contains at most $d$ distinct (non-$\bot$) values. In the following we assume without loss of generality that the rows of $A$ contain at most $d$ distinct entries each; otherwise we can replace the given instance~$(A, B, C)$ with one of the following instances: $(A^T, -C, -B)$, $(B, -C^T, -A^T)$, $(B^T, A^T, C^T)$, $(-C, B^T, -A)$ or $(-C^T, A, -B^T)$; in all cases there is an easily verifiable one-to-one correspondence of the exact triangles. It is easy to keep track of this correspondence to make all future outputs consistent with the original instance.

The algorithm can be split into two steps: In the first step we will ensure a weak form of degree-regularity (ensuring that each entry $A$ appears at most $\Order(\frac{n}{d})$ times in its respective row, and that the matrix $B$ is similarly constrained). In a second step, we will uniformize $A$, $B$ and $C$ then. 

\paragraph{Step 1: Outer Regularization.}
We partition $A$ into matrices $A^{(0)}, \dots, A^{(\ceil{\log n})}$ where the matrix~$A^{(x)}$ contains only entries from $A$ that appear at least $2^x$ times and at most $2^{x+1}$ in their respective row (the remaining entries are filled up with $\bot$). We similarly partition $B$ into $B^{(0)}, \dots, B^{(\ceil{\log n})}$ where the matrix~$B^{(y)}$ contains only entries from $B$ that appear at least $2^y$ times and at most $2^{y+1}$ times in their respective column. We then repeat the following steps for each pair $(x, y) \in [\ceil{\log n}]^2$; for ease of notation we write~$A$ to denote~$A^{(x)}$ and $B$ to denote $B^{(y)}$ in the following proof. Letting $d$ denote the maximum number of distinct entries in the rows of $A$, we clearly have that $d \cdot 2^x \leq n$.

We distinguish two cases: If $2^y \leq n / (d \Delta)$, then we will list all exact triangles in the instance in subcubic time. Specifically, enumerate all pairs $(i, j)$ and list the at most $d$ combinations to express~\makebox{$C[i, j] = a + b$} as the sum of an entry $a$ of the $i$-th row in $A$, and an entry $b$ in the $j$-th column of $B$. For each such pair $(a, b)$, further enumerate all indices~$k$ with~\makebox{$B[k, j] = b$}. As there are at most $2^{y+1} \leq 2 n / (d \Delta)$ such indices $k$, in total we enumerate at most~\makebox{$n^2 \cdot d \cdot 2n / (d \Delta) = \Order(n^3 / \Delta)$} triples. We include each exact triangle detected in this way into the set $T$ and terminate. It is easy to see that we miss no exact triangle with this approach.

In the other case, if $2^y > n / (d \Delta)$, we continue with the uniformization as follows. Note that now there can be at most $d' := d \Delta$ weights in each column of $B$.

\paragraph{Step 2: Uniformization.}
Let $X_i$ denote the non-$\bot$ entries in the $i$-th row of $A$, and let $Y_j$ denote the non-$\bot$ entries in the $j$-th column in $B$.
We apply \cref{lem:pop-sums-decomp} with the sets $X_1, \dots, X_n$ and $Y_1, \dots, Y_n$ of size at most $d'$ (indeed we have $|X_i| \leq d \leq d'$ by the initial assumption that each row in $A$ contains at most~$d$ entries, and $|Y_j| \leq d'$ by the last paragraph) and with parameter $\Delta' := \Delta^2$. This yields partitions
\begin{alignat*}{2}
    X_i &= X_{i, 1} \sqcup \dots \sqcup X_{i, \Delta^4} \sqcup X'_i \qquad &&\text{(for all $i \in [n]$),} \\
    Y_j &= Y_{j, 1} \sqcup \dots \sqcup Y_{j, \Delta^4} \sqcup Y'_j \qquad &&\text{(for all $j \in [n]$).}
\end{alignat*}
We will deal with the exact triangles in $(A, B, C)$ in two different cases. Let us call an exact triangle~$(i, k, j)$ \emph{exceptional} if $A[i, k] \in X'_i$ or $B[k, j] \in Y'_j$, and \emph{ordinary} otherwise. In the following steps we will first explicitly list all exceptional exact triangles, and then construct equivalent instances that preserve all ordinary exact triangles.
\begin{enumerate}
    \setlength\parskip{0pt plus 1pt}
    \item \emph{(The exceptional case)} We demonstrate how to list all exceptional exact triangles with $A[i, k] \in X_i'$; a symmetric argument shows how to list the exceptional exact triangles with $B[k, j] \in Y_j'$. Let $A'$ denote the matrix obtained from $A$ in which we keep all entries $A[i, k]$ with $A[i, k] \in X'_i$ and replace all other entries by $\bot$; then our goal is to list all exact triangles in the instance $(A', B, C)$. For each pair $(i, j) \in [n]^2$, we first approximate the set of popular sums in $X_i' + Y_j$ by calling \cref{lem:approx-popular-sums} with parameter $t := 2d' / \Delta'$ yielding sets $P_{4d' / \Delta'}(X_i', Y_j) \subseteq P_{i, j} \subseteq P_{2d' / \Delta'}(X_i', Y_j)$. Let us call an exceptional exact triangle $(i, k, j)$ \emph{popular} if $C[i, j] \in P_{i, j}$ and \emph{unpopular} otherwise. We then separately list the unpopular and popular triangles, respectively, as follows:
    \begin{enumerate}[label=1.\arabic*.]
        \item \emph{(The unpopular case)} We enumerate all pairs $(i, j) \in [n]^2$ with $C[i, j] \not\in P_{i, j}$. For each such pair, we enumerate all elements $a \in X_i'$, and test whether $C[i, j] - a \in Y_j$, thereby discovering all representations $C[i, j] = a + b$ for $(a, b) \in X_i' \times Y_j$. Recall that since $C[i, j] \not\in P_{i, j} \supseteq P_{4 d / \Delta}(X_i', Y_j)$ there are at most~$4d' / \Delta'$ such representations. For each representation we list the $\leq 2^{x+1} \leq 2n / d$ indices $k$ with $A[i, k] = a$ and record all triples $(i, k, j)$ detected this way. Note that in total we enumerate at most $n^2 \cdot 4d' / \Delta' \cdot 2n / d = \Order(n^3 / \Delta)$ triples in this step.
        \item \emph{(The popular case)} To list all exceptional popular exact triangles, we enumerate all $(i, j) \in [n]^2$ with $C[i, j] \not\in P_{i, j}$, and all nodes $k \in [n]$, and test whether $(i, k, j)$ in an exact triangle. By Property~(2) of \cref{lem:pop-sums-decomp} there are at most $n^2 / \Delta'$ pairs $(i, j)$ with a nonempty set of popular sums $P_{2d' / \Delta'}(X'_i, Y_j) \neq \emptyset$. Since also $P_{i, j} \subseteq P_{2d' / \Delta'}(X'_i, Y_j)$, in particular, we list at most $n^2 / \Delta'$ pairs $(i, j)$ in this step, leading to at most $\Order(n^3 / \Delta') = \Order(n^3 / \Delta)$ triples.
    \end{enumerate}
    
    \item \emph{(The ordinary case)} Let $S_g, T_h$ and $s_{i, g}, t_{j, h}$ (for $g, h \in [\Delta^4]$) denote the sets and shifts provided by Property (1) of \cref{lem:pop-sums-decomp}. We enumerate all pairs $g, h \in [\Delta^4]$, and construct matrices $A_g, B_h, C_{g, h}$ as follows:
    \begin{align*}
        A_g[i, k] &=
        \begin{cases}
            A[i, k] + s_{i, g} &\text{if $A[i, k] \in X_{i, g}$,} \\
            \bot &\text{otherwise,}
        \end{cases} \\
        B_h[k, j] &=
        \begin{cases}
            B[k, j] + t_{j, h} &\text{if $B[k, j] \in Y_{j, h}$,} \\
            \bot &\text{otherwise,}
        \end{cases} \\
        C_{g, h}[i, j] &= C[i, j] + s_{i, g} + t_{j, h}. \vphantom{\Big(}
    \end{align*}
    By this construction each triple $(i, k, j)$ is an exact triangle in $(A_g, B_h, C_{g, h})$ only if $(i, k, j)$ is an exact triangle in the original instance $(A, B, C)$. Moreover, each ordinary exact triangle $(i, k, j)$ in the original instance appears in exactly one instance $(A_g, B_h, C_{g, h})$ (namely the unique instance satisfying that $A[i, k] \in X_{i, g}$ and $B[k, j] \in Y_{j, h}$).

    Furthermore, the matrix $A_g$ contains at most $d'$ distinct entries, namely the elements of $S_g$ (indeed, by Property (1) of \cref{lem:pop-sums-decomp} each non-$\bot$ entry satisfies that $A_g[i, k] = A[i, k] - s_{i, g} \in X_{i, g} - s_{i, g} \subseteq S_g$). Similarly, the matrix $B_h$ only contains the at most $d'$ entries in $T_h$. 
    
    All in all, we have ensured that the matrices $A_g$ and $B_h$ are already uniform, and it remains to uniformize~$C_{g, h}$. We achieve this by a similar popular-versus-unpopular distinction as before. Specifically, we compute the set $P := P_{d / \Delta^9}(S_g, T_h)$ (by brute-force in time $\Order((d')^2)$, say), and call an ordinary exact triangle $(i, k, j)$ \emph{popular} if $C_{g, h}[i, j] \in P$ and \emph{unpopular} otherwise. In the following two steps we will deal with the unpopular and popular triangles, respectively:
    \begin{enumerate}[label=2.\arabic*.]
        \item \emph{(The unpopular case)} We can again list all ordinary unpopular exact triangles by enumerating all pairs $(i, j) \in [n]^2$ with $C_{g, h} \not\in P$, and all representations $C_{g, h}[i, j] = a + b$ for $(a, b) \in S_j \times T_j$; there can be at most $d / \Delta^9$ many such representations by the definition of $P$. For each such representation we enumerate all indices $k$ with $A[i, k] = a$, test whether the triple~$(i, k, j)$ forms an exact triangle, and include the triple into $T$ in that case. Per instance $(A_g, B_h, C_{g, h})$ we thus enumerate at most $n^2 \cdot d / \Delta^9 \cdot 2^{x+1} = \Order(n^3 / \Delta^9)$ triples, and across the $\Delta^8$ instances~\smash{$(A_g, B_h, C''_{g, h})$} we enumerate at most $\Order(n^3 / \Delta)$ triples.
        \item \emph{(The popular case)} Finally, let $C_{g, h}'$ denote the matrix obtained from $C_{g, h}$ in which we replace all entries that do not appear in $P$ by $\bot$. Consider the Exact Triangle instance~\smash{$(A_g, B_h, C_{g, h}')$}; we claim that it is $(d \cdot \Delta^{11})$-uniform. We have argued before that $A_g$ and $B_h$ contain at most~\makebox{$d' = d \cdot \Delta$} distinct entries, and the number of distinct entries in $C_{g, h}'$ is indeed at most
        \begin{equation*}
            |P| \leq \frac{|S_g| \cdot |T_h|}{d / \Delta^9} \leq \frac{d' \cdot d'}{d / \Delta^9} = \frac{d^2 \cdot \Delta^2}{d / \Delta^9} = d \cdot \Delta^{11}.
        \end{equation*}
        In a final step we reduce the $(d \cdot \Delta^{11})$-uniform instance~\smash{$(A_g, B_h, C_{g, h}')$} to $\Delta^{33}$ equivalent $d$-uniform instances by \cref{obs:exact-tri-uniformization-naive}; these become the instances claimed in the lemma statement.
    \end{enumerate}    
\end{enumerate}

\paragraph{Correctness.}
Concerning the correctness, recall that we have made sure that each exact triangle was either listed and included in $T$ or remained in an instance~\smash{$(A_g, B_h, C_{g, h}')$}, proving Property~(i). Property~(ii) also follows from the previous description. It remains to prove Property~(iii), i.e., to bound the number of constructed instances. Recall that in step 1 we create $\Order(\log^2 n)$ initially instances. Each such instance is transformed into $\Delta^8$ instances~\smash{$(A_g, B_h, C_{g, h}')$}, which, in turn, are transformed into $\Delta^{33}$ instances each. Therefore, we indeed construct $\Order(\Delta^{41} \log^2 n) = \poly(\Delta \log n)$ instances. The bound on $|T|$ follows from the in-text descriptions.

\paragraph{Running Time.}
The running time of the initial regularization step is $\widetilde\Order(n^2 + n^3 / \Delta)$. In the uniformization step, we have first applied \cref{lem:pop-sums-decomp} which runs in deterministic time $n^{2+\order(1)} d \Delta$ (using that all numbers are polynomially bounded). Then, in case 1 we issue $n^2$ calls to \cref{lem:approx-popular-sums} each running in deterministic time $(d')^2 / (d' / \Delta') \cdot n^{\order(1)} = d' \Delta' \cdot n^{\order(1)} = d \Delta^3 \cdot n^{\order(1)}$. The running time of all brute-force steps~1.1,~1.2 and~2.1 can be bounded as follows: Up to an overhead of $\Order(n^2 d)$, we spend constant time per enumerate triple. Recall that we enumerate~\smash{$\widetilde\Order(n^3 / \Delta)$} triples in total, leading to a running time of~\smash{$\widetilde\Order(n^3 / \Delta)$}. Finally, in step 2.2, we spend time $\widetilde\Order(n^2 \Delta^{41})$ to prepare all instances. Summing over all contributions, the running time is~\smash{$\widetilde\Order(n^{2+\order(1)} d \poly(\Delta) + n^3 / \Delta)$} as claimed.
\end{proof}

\subsection{Regularization} \label{sec:apsp-few-weights:sec:regularization}

\begin{observation}[Naive Regularization] \label{obs:exact-tri-regularization-naive}
Let $d, r, R \geq 1$, and let $(A, B, C)$ be a $d$-uniform $rR$-regular Exact Triangle instance. We can compute $d$-uniform $r$-regular Exact Triangle instances~\smash{$\set{(A_\ell, B_\ell, C_\ell)}_{\ell=1}^{R^6}$} such that $\mathcal T(A, B, C) = \bigsqcup_{\ell \in [R^6]} \mathcal T(A_\ell, B_\ell, C_\ell)$ in time $\Order(n^2 R^6)$.
\end{observation}
\begin{proof}
First partition $A$ into $R$ matrices $A_1, \dots, A_R$ such that in each matrix $A_i$ each entry appears at most $r$ times in its respective row. Then further partition each matrix $A_i$ into matrices $A_{i, 1}, \dots, A_{i, R}$ such that each entry appears at most $r$ times in its respective column. We similarly partition $B$ and $C$, and then consider all instances $(A_{i, j}, B_{k, \ell}, C_{g, h})$ for $i, j, k, \ell, g, h \in [R]$. Clearly each exact triangle appears in the original instance $(A, B, C)$ appears in exactly one such instance.
\end{proof}

\begin{lemma}[Regularization] \label{lem:exact-tri-regularization}
Let $\epsilon > 0$ be a constant, let $\Delta \geq 1$, and let $(A, B, C)$ be a $d$-weights Exact Triangle instance. Then we can compute Exact Triangle instances $\set{(A_m, B_m, C_m)}_{m=1}^M$ and a set of triples $T \subseteq [n]^3$ such that:
\begin{enumerate}[label=(\roman*)]
    \item $\mathcal T(A, B, C) = T \sqcup \bigsqcup_{m \in [M]} \mathcal T(A_m, B_m, C_m)$.
    \item For each $m \in [M]$, the instance $(A_m, B_m, C_m)$ is $d'$-uniform and $\frac{n}{d'}$-regular for some $d' \leq d$.
    \item \smash{$M \leq d^\epsilon \cdot \Delta^{2^{\Order(1 / \epsilon)}} (\log n)^{\Order(1 / \epsilon^2)}$} and $|T| \leq \widetilde\Order(n^3 / \Delta)$.
\end{enumerate}
The algorithm runs in time~\smash{$\widetilde\Order(n^2 d \cdot \Delta^{2^{\Order(1 / \epsilon)}} (\log n)^{\Order(1 / \epsilon^2)} + n^3 / \Delta)$}.
\end{lemma}
\begin{proof}
Fix $\rho = d^{\epsilon/6}$, and consider the recursive algorithm summarized in \cref{alg:regularization}. The algorithm takes as input a $d$-weights Exact Triangle instance (i.e., with the promise that the rows or the columns in at least one of the matrices~$A, B, C$ each contain at most $d$ distinct entries). The output is a set of exact triangles~$T$, and a set of equivalent Exact Triangle instances each promised to be $d'$-uniform and $(\frac{n}{d'} \cdot \rho)$-regular for some~$d' \leq d$. Note that this is worse by a factor of $\rho$ than stated in the lemma statement---we can fix this in a post-processing step by \cref{obs:exact-tri-regularization-naive} splitting each returned instance into $\rho^6$ instances that are $d'$-uniform and $\frac{n}{d'}$-regular.

\begin{algorithm}[t]
\caption{Implements the algorithm from \cref{lem:exact-tri-regularization}: The input is a $d$-weights Exact Triangle instance $(A, B, C)$ and a parameter $\Delta$. The output is a set of exact triangles $T$, and a set of Exact Triangle instances each promised to be $d'$-uniform and $(\frac{n}{d'} \cdot \rho)$-regular for some $d' \leq d$. Throughout, $\rho$ is a global parameter.} \label{alg:regularization}
\begin{enumerate}[itemsep=0pt]
    \item If $d = 0$ then stop.
    \item Call \cref{lem:exact-tri-uniformization} on $(A, B, C)$ with parameters $d, \Delta$ to compute $\set{(A_\ell, B_\ell, C_\ell)}_{\ell=1}^L$ and $T$. Output all triples in $T$.
    \item For all $\ell \in [L]$:
    \begin{enumerate}[topsep=0pt, itemsep=0pt, label=3.\arabic*.]
        \item Partition the non-$\bot$ entries of $A_\ell$ into three matrices $A_\ell^{\text{row}}, A_\ell^{\text{col}}, A_\ell^{\text{reg}}$ such that $A_\ell^{\text{row}}$ contains only entries that appear more than $\frac{n}{d} \cdot \rho$ times in their row, $A_\ell^{\text{col}}$ contains only entries that appear more than $\frac{n}{d} \cdot \rho$ times in their column, and $A_\ell^{\text{reg}}$ contains only entries that appear at most $\frac{n}{d} \cdot \rho$ in their row and column. Similarly, partition $B_\ell$ into $B_\ell^{\text{row}}, B_\ell^{\text{col}}, B_\ell^{\text{reg}}$ and $C_\ell$ into $C_\ell^{\text{row}}, C_\ell^{\text{col}}, C_\ell^{\text{reg}}$.
        \item Recurse on $(A_\ell^{\text{row}}, B_\ell, C_\ell)$, $(A_\ell^{\text{col}}, B_\ell, C_\ell)$, $(A_\ell, B_\ell^{\text{row}}, C_\ell)$, $(A_\ell, B_\ell^{\text{col}}, C_\ell)$, $(A_\ell, B_\ell, C_\ell^{\text{row}})$,\newline and~$(A_\ell, B_\ell, C_\ell^{\text{col}})$, each with parameters $d \gets \floor{d / \rho}$ and $\Delta \gets \Delta \cdot 12L$.
        \item Output the instance $(A_\ell^{\text{reg}}, B_\ell^{\text{reg}}, C_\ell^{\text{reg}})$
    \end{enumerate}
\end{enumerate}
\end{algorithm}

In the trivial base case $d = 0$ there can be no non-$\bot$ entry in at least one of the matrices and hence the instance does not contain any solution. Suppose that $d \geq 1$ in the following. We apply \cref{lem:exact-tri-uniformization} to uniformize the given instance, i.e., to obtain instances $\set{(A_\ell, B_\ell, C_\ell)}_{\ell=1}^L$ and a set of remaining triples~$T$ that capture all exact triangles in the original instance. We report the triples in $T$. Then we focus on each instance $(A_\ell, B_\ell, C_\ell)$ separately. We partition the matrix $A_\ell$ as follows into three matrices $A_\ell^{\text{row}}$, $A_\ell^{\text{col}}$, and~$A_\ell^{\text{reg}}$. Each non-$\bot$ entry appears in exactly one of these three matrices such that
\begin{itemize}[itemsep=0pt]
    \item each non-$\bot$ entry in $A_\ell^{\text{row}}$ appears more than $\frac{n}{d} \cdot \rho$ times in its respective \emph{row} in $A_\ell$,
    \item each non-$\bot$ entry in~$A_\ell^{\text{col}}$ appears more than $\frac{n}{d} \cdot \rho$ times in its respective \emph{column} in $A_\ell$, and
    \item each non-$\bot$ entry in $A_\ell^{\text{reg}}$ appears at most~$\frac{n}{d} \cdot \rho$ times in its respective \emph{row and column} in $A_\ell$.
\end{itemize}
Such a partition can e.g.\ be achieved by first putting the entries that appear too often in their rows into~$A_\ell^{\text{row}}$ and deleting them after. Then repeat the same for the columns with $A_\ell^{\text{col}}$ and put the remaining entries into~$A_\ell^{\text{reg}}$. We similarly partition $B_\ell$ into $B_\ell^{\text{row}}, B_\ell^{\text{col}}, B_\ell^{\text{reg}}$ and $C_\ell$ into $C_\ell^{\text{row}}, C_\ell^{\text{col}}, C_\ell^{\text{reg}}$. The insight is that each of the six instances
\begin{equation*}
    (A_\ell^{\text{row}}, B_\ell, C_\ell), (A_\ell^{\text{col}}, B_\ell, C_\ell), (A_\ell, B_\ell^{\text{row}}, C_\ell), (A_\ell, B_\ell^{\text{col}}, C_\ell), (A_\ell, B_\ell, C_\ell^{\text{row}}), (A_\ell, B_\ell, C_\ell^{\text{col}}),
\end{equation*}
is $\floor{d / \rho}$-weights (since, e.g., in each row in $A_\ell^{\text{row}}$ there can be at most $d / \rho$ distinct entries). We will therefore recur on these six instances with parameter $d \gets \floor{d / \rho}$ and $\Delta \gets \Delta \cdot 12L$ (this choice for $\Delta$ will become clear later). It remains to consider the exact triangles that appear in $(A_\ell^{\text{reg}}, B_\ell^{\text{reg}}, C_\ell^{\text{reg}})$. But by construction this instance is $d$-uniform (by Property (i) of \cref{lem:exact-tri-uniformization}) and $(\frac{n}{d} \cdot \rho)$-regular. We can thus count $(A_\ell^{\text{reg}}, B_\ell^{\text{reg}}, C_\ell^{\text{reg}})$ towards the instances constructed by the algorithm.

\paragraph{Correctness of Property (i).}
The proof of Property (i) is straightforward by induction, verifying that we consider each exact triangle $(i, k, j)$ either in the set of exceptional triples $T$ (in step 2), or in a recursive call (in step 3b), or in the reported instances (in step 3c).

\paragraph{Correctness of Property (ii).}
As stated before, any instance $(A_\ell^{\text{reg}}, B_\ell^{\text{reg}}, C_\ell^{\text{reg}})$ reported by the algorithm is $d$-uniform and $(\frac{n}{d} \cdot \rho)$-regular (which by an application of \cref{obs:exact-tri-regularization-naive} becomes $\frac{n}{d}$-regular). In each recursive call our value for $d$ only decreases, and so all recursive calls satisfy the same property for some~$d' \leq d$. 

\paragraph{Correctness of Property (iii).}
We first bound the recursion depth of the algorithm. Note that with each recursive call we decrease $d$ by a factor of $\rho = d^{\epsilon / 6}$. Thus, the recursion depth is $H := \ceil{\log_\rho(d)} = \Order(1 / \epsilon)$. Let $M_h$ denote the number of recursive calls recursion tree at depth $h$, and let $\Delta_h$ and $L_h$ denote the value of the parameters $\Delta$ and $L$ at recursion depth $h$. Note that
\begin{align*}
    \Delta_0 &= \Delta, \\
    \Delta_{h+1} &= \Delta_h \cdot 12L_h = \Delta_h^{\Order(1)} (\log n)^{\Order(1)},
\end{align*}
(where the last bound is due to Property (iii) of \cref{lem:exact-tri-uniformization}). By induction, it follows that
\begin{equation*}
    \Delta_h = \Delta^{2^{\Order(h)}} (\log n)^{\Order(h)}.
\end{equation*}
Moreover, we have that
\begin{align*}
    M_0 &= 1, \\
    M_{h+1} &= M_h \cdot 6 L_h,
\end{align*}
and thus
\begin{equation*}
    M_h = \prod_{h' = 0}^h \parens*{\Delta^{2^{\Order(h')}} (\log n)^{\Order(h')}} = \Delta^{2^{\Order(h)}} (\log n)^{\Order(h^2)}.
\end{equation*}
In particular, at the deepest level of the recursion we have spawned at most $\Delta^{2^{\Order(1/\epsilon)}} (\log n)^{\Order(1/\epsilon^2)}$ recursive calls. Each such recursive call leads to $\rho^6 \leq \Order(d^{\epsilon})$ instances (by the final application of \cref{obs:exact-tri-regularization-naive}), confirming the claimed bound on $M$.

\paragraph{Running Time.}
The running time is dominated by calling the uniformization decomposition at each recursive call:
\begin{equation*}
    \sum_{h=0}^H M_h \cdot \widetilde\Order\parens*{n^{2+\order(1)} d \poly(\Delta_h) + \frac{n^3}{\Delta_h}} = 
    n^{2+\order(1)} \cdot d \cdot \Delta^{2^{\Order(1/\epsilon)}} \cdot (\log n)^{\Order(1/\epsilon^2)} + \widetilde\Order\parens*{\sum_{h=0}^H \frac{n^3 M_h}{\Delta_h}}.
\end{equation*}
The above recursive expressions yield that~\smash{$M_0 / \Delta_0 = 1 / \Delta$} and~\smash{$M_{h+1} / \Delta_{h+1} = \frac{1}{2} \cdot M_h / \Delta_h$}. Thus, we can bound the sum in the second term of the running time by
\begin{equation*}
    \sum_{h=0}^H \frac{M_h}{\Delta_h} = \frac{1}{\Delta} \sum_{h=0}^H \frac{1}{2^h} \leq \frac{2}{\Delta}.
\end{equation*}
The stated time bound follows.
\end{proof}

This concludes the proof of the regularization step, and readies us to complete the proof of \cref{thm:exact-tri-few-weights} by combining \cref{lem:exact-tri-uniform-regular,lem:exact-tri-regularization}.

\begin{proof}[Proof of \cref{thm:exact-tri-few-weights}]
Let $\delta > 0$ and let $d = n^{{3-\omega}-\delta}$. 
Let $\epsilon > 0$ and $\Delta \geq 1$ be parameters to be determined. We apply \cref{lem:exact-tri-regularization} to reduce a given $d$-weights Exact Triangle instance to
\begin{equation*}
    M = d^\epsilon \cdot \Delta^{2^{\Order(1/\epsilon)}} (\log n)^{\Order(1/\epsilon^2)}    
\end{equation*}
instances that are each $d'$-uniform and \smash{$\frac{n}{d'}$}-regular, for some $d' \leq d$. Using \cref{lem:exact-tri-uniform-regular} we solve each such instance in time $\widetilde\Order(n^{3-(3-\omega)/7} (d')^{1/7}) = \widetilde\Order(n^{3-(3-\omega)/7} d^{1/7}) = \widetilde\Order(n^{3-\delta/7})$. Taking also the running time of \cref{lem:exact-tri-regularization} into account, the total running time is
\begin{equation*}
    \widetilde\Order(n^{3-\delta/7}\cdot d^\epsilon \cdot \Delta^{2^{\Order(1/\epsilon)}} (\log n)^{\Order(1/\epsilon^2)} + n^{2+\order(1)} \cdot d \cdot \Delta^{2^{\Order(1/\epsilon)}} (\log n)^{\Order(1/\epsilon^2)} + n^3 / \Delta).
\end{equation*}
Picking $\epsilon = \frac{\delta}{14}$ and $\Delta = n^{2^{-c/\epsilon}}$ for some sufficiently large constant $c$, this becomes $\Order(n^{3-\epsilon'})$ time (for some constant~\makebox{$\epsilon' > 0$} which is exponentially small in $1/\delta$). 
\end{proof}

\begin{remark} \label{rem:eps-dependence}
For $d$-weights Exact Triangle where $d=n^{3-\omega-\delta}$, the proof above yields time complexity $O(n^{3-\epsilon'})$ where $\epsilon'$ is exponentially small in $1/\delta$. Here we sketch a small modification that can improve $\epsilon'$ to depend polynomially on $\delta$.  Further improving to $\epsilon' \ge \Omega(\delta)$ seems challenging (at least within our current recursion-based framework), and we leave it as an interesting open problem.

Recall that in our recursion algorithm (\cref{alg:regularization}) for $d$-weights Exact Triangle, the value of $d$ decreases at each recursive step, and we increase the parameter $\Delta$ accordingly, but keep the parameter $\rho$ unchanged throughout all levels of recursion. 
The modification we make here is to also vary the value of $\rho$. 
More specifically, let $T(d)$ denote the total running time for solving a $d$-weights Exact Triangle instance (including both the time of running \cref{alg:regularization} to produce the uniform and regular instances, and the time of invoking \cref{lem:exact-tri-uniform-regular} to solve these instances).
Then, by previous discussions, $T(d)$ satisfies the following recurrence (where we are free to choose parameters $\rho,\Delta\ge 1$): \[T(d)\leq \widetilde O\left ( 6L\cdot
T\left(\frac{d}{\rho}\right)+n^{2+o(1)}d\Delta^{O(1)}+\frac{n^3}{\Delta}+L\rho^6n^{3-(3-\omega)/7}d^{1/7}\right ), \text{ where }L \le \poly(\Delta \log n),
\] since for each of the $L$ $d$-uniform instances produced by \cref{lem:exact-tri-uniformization}, we recurse on six
($d/\rho$)-weights instances, and we also call the naive
regularization (\cref{obs:exact-tri-regularization-naive}) and obtain $\rho^6$ $d$-uniform
and ($n/d$)-regular instances, each to be solved using \cref{lem:exact-tri-uniform-regular} in time $\widetilde O(n^{3-(3-\omega)/7} d^{1/7})$. 
Recall that at the top level, $d = n^{3-\omega-\delta}$ where $\delta>0$.
At the current level of recursion, let $\epsilon \ge \delta$ be such that $d = n^{3-\omega-\epsilon}$.
We now inductively show that $T(d) = T(n^{3-\omega-\epsilon})\leq cn^{3-\epsilon^C/K}$ holds for some constants $c, C,K\ge 1$ to be determined later.

For the base case, if $\epsilon\ge 0.1$, then we can use our original proof of \cref{thm:exact-tri-few-weights} to solve $n^{3-\omega-0.1}$-weights Exact Triangle in $O(n^{3-\epsilon''})$ time for some absolute constant $\epsilon''>0$, so by setting $K\ge 1/\epsilon''$ and $c$ large enough we immediately have $T(n^{3-\omega-\epsilon}) \le cn^{3-\epsilon^C/K}$ as claimed.
Now we assume $\epsilon<0.1$. In the recurrence above, we set parameters
$\rho = n^{r_\epsilon}$ and $\Delta =
n^{s_\epsilon}$, where $s_\epsilon = 2\epsilon^C$ and $r_\epsilon = \epsilon/100$.
Assume the inductive hypothesis that
$T(d/\rho) = T(n^{3-\omega-(\epsilon+r_\epsilon)})\leq
cn^{3-(\epsilon+r_\epsilon)^C/K}$. 
Substituting, the recurrence becomes
\[T(n^{3-\omega-\epsilon})\leq
\widetilde O\left (cn^{O(s_\epsilon)+3-(\epsilon+r_\epsilon)^C/K}+n^{3-\epsilon+O(s_\epsilon)}+n^{3-s_\epsilon}+n^{O(s_\epsilon)+6r_\epsilon+3-(3-\omega)/7+(3-\omega-\epsilon)/7} \right ),
\] where we use that $L = \text{poly}(\Delta\log n) = n^{O(s_\epsilon)}$. In
order to show that $T(n^{3-\omega-\epsilon})\leq cn^{3-\epsilon^C/K}$,
it suffices to show that each summand in the right-hand-side of this
recurrence is at most $cn^{3-\epsilon^C/K}/q$, where $q\le \polylog(n)$ is four times the factors hidden by the $\widetilde O(\cdot)$. 
Since $\log_n q = o(1)$, we can assume $\log_n q \le \delta^{C} \le \epsilon^{C}$ by assuming $n\ge n_0$ for some large constant $n_0$ depending on $\delta$ and $C$; otherwise when $n<n_0$, we have a constant-time algorithm, and we can set constant $c\ge 1$ large enough in terms of $n_0$ to cover this.
Letting $k$ denote the absolute constant hidden by the $O(s_{\eps})$ in the recurrence, we have that \begin{itemize}
\item $cn^{O(s_\epsilon)+3-(\epsilon+r_\epsilon)^C/K}\leq cn^{3-\epsilon^C/K}/q$
  is satisfied when $(2kK+K+1) \epsilon^C \leq (1.01\epsilon)^C$ (where we used $\log_n q \le \epsilon^C$), which is satisfied by choosing $C$ large enough in terms of $k$ and $K$;

\item $n^{3-\epsilon+O(s_\epsilon)}\leq cn^{3-\epsilon^C/K}/q$ is satisfied when
  $\epsilon \ge k\cdot 2\epsilon^C + \epsilon^C/K+\epsilon^C$ (where we used $c\ge 1$), which is satisfied by
  choosing $C$ large enough in terms of $k$ (and using $\epsilon <0.1$);

\item $n^{3-s_\epsilon}\leq cn^{3-\epsilon^C/K}/q$ is satisfied when
  $(2-1/K)\epsilon^C\geq\log_n q$, which holds by our assumption that $\log_n q \le \epsilon^C$;
  and

\item $n^{O(s_\epsilon)+6r_\epsilon+3-(3-\omega)/7+(3-\omega-\epsilon)/7}\leq
  cn^{3-\epsilon^C/K}/q$ is satisfied when
  $(2k+1/K)\epsilon^C+\log_n q\leq(1/7 - 0.06)\epsilon$,
  which is similarly satisfied by choosing $C$ large enough in terms of $k$ and using $\epsilon <0.1$.
\end{itemize}

Thus, we have $T(n^{3-\omega-\epsilon})\leq cn^{3-\epsilon^C/K}$ for some
constants $c, C,K\ge 1$ as claimed. 
\end{remark}
\section{Hardness and Reductions}
\label{sec:hardness}
In this section our goal is to present conditional hardness for Node-Weighted APSP and to relate the problem to Min-Plus product for matrices with a small number of distinct weights. Throughout this section we refer to the Min-Plus Product problem of size $n_1 \times n_2 \times n_3$ with weights bounded by $U$ as the $(n_1, n_2, n_3 \mid U)$-Min Plus Product problem.

Chan, Vassilevska W. and Xu \cite{ChanWX21} study the directed unweighted APSP (u-dir-APSP) problem and provide a wide range of reductions and equivalences for it. The fastest known running time for u-dir-APSP is by Zwick \cite{Zwick02} and it is $\widetilde{O}(n^{2+\mu})$ where $\mu$ is defined to be the value s.t. $\omega(1,\mu,1)=1+2\mu$. If $\omega=2$, then $\mu=0.5$, and the current best bound on $\mu$ is $\mu<0.5277$~\cite{AlmanDVXXZ25}. 

Chan, Vassilevska W., and Xu \cite{ChanWX21} show that u-dir-APSP requires $n^{2+\mu-o(1)}$ time if and only if many other problems require $n^{2+\mu-o(1)}$ time as well, such as All-Pairs Longest Paths in unweighted DAGs and to $(n, n^\mu, n \mid n^{1-\mu})$-Min-Plus Product. In a later paper~\cite{ChanWX23} the same authors formulate the following hypothesis:

\begin{hypothesis}[u-dir-APSP]
In the word-RAM model with $O(\log n)$ bit words, u-dir-APSP on $n$-node graphs cannot be solved in time $O(n^{2+\mu-\epsilon}$, for any constant $\epsilon > 0$.
\end{hypothesis}

That is, the u-dir-APSP Hypothesis postulates that Zwick's algorithm is optimal. Another related hypothesis is the Strong APSP hypothesis, introduced in the same paper~\cite{ChanWX23}:

\begin{hypothesis}[Strong APSP Hypothesis]
In the word-RAM model with $O(\log n)$ bit words, APSP with weights bounded by $n^{3-\omega}$ cannot be solved in time $n^{3-\epsilon}$, for any constant $\epsilon > 0$.
\end{hypothesis}

Again, this hypothesis is equivalent to the assumption that the $(n, n, n \mid n^{3-\omega})$-Min Plus Product problem requires time $n^{3-o(1)}$.

\subsection{Hardness for \texorpdfstring{\boldmath$d$}{d}-Distinct Weights APSP}
We formulate an interpolation between the u-dir-APSP and the Strong APSP Hypotheses: 

\begin{hypothesis}[Bounded Min-Plus Product]
For any $\gamma \in [1/2,1]$, in the word-RAM model with $O(\log n)$ bit words, $(n, n^\gamma, n \mid n^\gamma)$-Min-Plus cannot be solved in time $O(n^{2+\gamma-\epsilon})$, for any constant $\epsilon > 0$.
\end{hypothesis}

The u-dir-APSP Hypothesis is equivalent to the case of $\gamma=1/2$ (if $\omega = 2$) and the Strong APSP Hypothesis is equivalent to the case of $\gamma=1$ (if $\omega = 2$). In words, we postulate that for every $\gamma$ between~$1/2$ and~$1$, Min-Plus product also requires essentially brute-force time. The statement is consistent with the known algorithms for Min-Plus product and we find it just as believable as the Strong APSP and u-dir-APSP Hypotheses. In fact, Chan, Vassilevska W. and Xu \cite{ChanWX23} prove that if $\omega = 2$ the u-dir-APSP Hypothesis implies lower bounds for Min-Plus Product of all sizes where the inner dimension is strictly smaller than the outer dimensions:

\begin{lemma}[{{\cite[Corollary 7.6]{ChanWX23}}}]
If $\omega = 2$, the u-dir-APSP Hypothesis and the Bounded Min-Plus Product Hypothesis for any fixed $\gamma \in [1/2, 1)$ are equivalent.
\end{lemma}

We give a reduction from $(n, n^{1/2+\eps}, n \mid n^{1/2+\eps})$-Min-Plus to APSP in graphs with at most $n^{2\eps}$ distinct weights, obtaining the theorem below:

\begin{theorem}\label{thm:dapsp-merge1}
Under the Bounded Min-Plus Product Hypothesis, for every $d=n^\gamma\geq 2$ with $\gamma\in [0,1]$, APSP in (directed or undirected) graphs with at most $d$ distinct weights and hence also $d$-weights APSP requires $\sqrt{d}n^{2.5-o(1)}$ time in the word-RAM model.
\end{theorem}

\begin{proof}
Let $A$ and $B$ be the given matrices, where $A$ is $n\times n^{1/2+\eps}$ and $B$ is $n^{1/2+\eps}\times n$ and where all entries of $A$ and $B$ are integers in $[0,n^{1/2+\eps})$, where $\eps \in [0,1/2]$. As usual, we will assume that $n^{1/2+\eps}$ is an integer.

Let $q=\ceil{n^{1/2-\eps}}$. For any integer $x\in [0,n^{1/2+\eps})$ we can represent it as $x=x_1\cdot q + x_2$ where $x_2=x\bmod q \in \{0,\ldots,q-1\}$ and $x_1=\lfloor x/q\rfloor \in \{0,\ldots,n^{2\eps} - 1\}$. We define four matrices $A',A'',B,B''$ as follows, for all $i,j\in [n],k\in [n^{1/2+\eps}]$:
\begin{alignat*}{2}
    A'[i,k] &= q\cdot \left\lfloor \frac{A[i,k]}{q} \right\rfloor, \qquad & A''[i,k] &= A[i,k] \bmod q, \\
    B'[k,j] &= q\cdot \left\lfloor \frac{B[k,j]}{q} \right\rfloor, \qquad & B''[k,j] &= B[k,j] \bmod q.
\end{alignat*}
Notice that all entries in $A'$ and $B'$ lie in $q\cdot \{0,1,\dots,n^{2\eps}-1\}$ and in particular there are at most $n^{2\eps}$ distinct weights that appear as entries in $A'$ and $B'$.

Now we build a directed graph $G$ on $O(n)$ vertices as follows (following a similar construction from \cite{ChanWX21}). We will later make the graph undirected. We create a node set $R$ consisting with a node $r_i$ for every $i\in [n]$, corresponding to row $i$ of $A$. Similarly, we create a node set $C$ consisting with a node $c_j$ for every $j\in [n]$, corresponding to column $j$ of $A$. For every $k\in [n^{1/2+\eps}]$ we create a path $P_k$ on $2q-1=2n^{1/2-\eps}-1$ vertices with a middle vertex $w_k$ as follows:
\[P_k:=x_{k,q-1}\rightarrow x_{k,q-2}\rightarrow\ldots\rightarrow x_{k,1}\rightarrow x_{k,0}=w_k=y_{k,0}\rightarrow y_{k,1}\rightarrow \ldots \rightarrow y_{k,q-2}\rightarrow y_{k,q-1}.\]

Notice that the number of vertices in the graph is $O(n)$ due to our choice of $q$. All the edges within the paths $P_k$ have weight $1$. Now, for every $i\in [n],k\in [n^{1/2+\eps}]$, we add an edge of weight $q\cdot\lfloor A[i,k] / q\rfloor$ from $r_i$ to $x_{k,b}\in P_k$ where $b=A[i,k] \bmod q$. Similarly, for every $j\in [n],k\in [n^{1/2+\eps}]$, we add an edge of weight $q\cdot \lfloor B[k,j] / q \rfloor$ from $y_{k,b'}\in P_k$ to $c_j$ where $b'=B[k,j] \bmod q$. See Figure \ref{fig:distinctred}.

\begin{figure}[t]
\centering
\includegraphics[width=10cm]{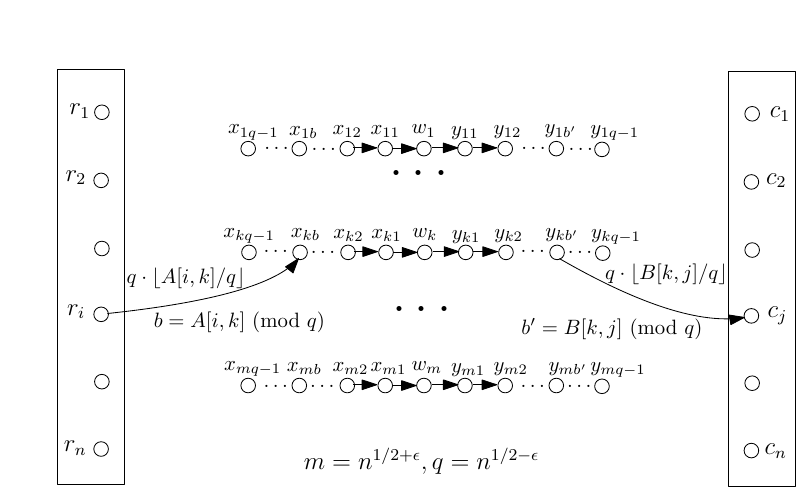}
\caption{Our reduction from $(n, n^{1/2+\eps}, n \mid n^{1/2+\eps})$-Min Plus Product to directed APSP with $n^{2\eps}$ distinct weights.}
\label{fig:distinctred}
\end{figure}

Every directed path from $r_i\in R$ to $c_j\in C$ has the following form: for some $k\in [n^{1/2+\eps}]$, take an edge $(r_i,x_{k,b})$ of weight $q\cdot \lfloor A[i,k] / q \rfloor$, followed by a $b+b' = [A[i,k]\bmod q]+[B[k,j]\bmod q]$ weight (sum over the integers) subpath of $P_k$ to $y_{k,b'}$ and then an edge to $c_j$.

The weight of this path is exactly $A[i,k]+B[k,j]$ and thus the shortest path from $r_i$ to $c_j$ is $\min_k A[i,k]+B[k,j]$.

The graph has $\leq n^{2\eps} + 1$ distinct weights. Thus, if APSP in graphs with at most $n^\gamma \ge 2$ distinct weights can be solved in $O(n^{2.5+\gamma/2-\delta})$ time for some $\delta>0$, then we can compute the Min-Plus product of $A$ and $B$ in time $O(n^{2.5+\eps-\delta})$, refuting the Bounded Min-Plus Hypothesis.

To get the same result for undirected graphs, notice that if we can ensure that the shortest path between $r_i$ and $c_j$ (for every $i,j\in [n]$) when we disregard the edge directions in our reduction uses at most one edge out of $R$ and at most one edge into $C$, then this path will intersect a single path $P_k$ and will just follow its subpath without revisiting any edge more than once. Thus as long as we can ensure that the shortest $r_i$ to $c_j$ undirected path uses exactly two edges incident to $R\cup C$, the reduction would also work for undirected graphs.

We can do this by adding a very large weight $M$ to every edge incident to $R\cup C$. This keeps the number of distinct edge weights in the graph the same. $M=2n^{1/2+\eps}$ suffices.
\end{proof}

In particular, the above theorem implies that APSP in {\em undirected} graphs with $2$ distinct weights ($\{2\sqrt n,1\}$) requires $n^{2.5-o(1)}$ time under the Bounded Min-Plus Hypothesis, and even under the u-dir-APSP Hypothesis. Notably, undirected APSP with one distinct weight is in $\widetilde{O}(n^\omega)$ time via Seidel's algorithm \cite{Seidel95}, so that there is a big jump in complexity with two distinct weights. This completes half of the proof of \cref{thm:dapsp}; the other half is proved in the upcoming \cref{lem:merge2}.

\subsection{Hardness and Reductions for Node-Weighted APSP}
Under the u-dir-APSP Hypothesis, node-weighted APSP in {\em directed} graphs must also require $n^{2.5-o(1)}$ time, and so our $\widetilde{O}(n^{(3+\omega)/2})$ time algorithm is conditionally optimal if $\omega=2$.
If $\omega>2$, however, there is a difference between our running time and the \smash{$\widetilde{O}(n^{2+\mu})$} running time of u-dir-APSP. Furthermore, unweighted {\em undirected} APSP can be solved significantly faster than u-dir-APSP, in $\widetilde{O}(n^\omega)$ time \cite{Seidel95}, whereas undirected node-weighted APSP seems to be just as difficult as directed node-weighted APSP: our techniques fail to achieve a better running time.

Here we will provide a hardness reduction that relates node-weighted APSP to Min-Plus Product problem with a small number of distinct weights (per row or column), building a connection with the second part of our paper.

We begin with two definitions:
\begin{definition}[$d$-Column Weight $(n,D)$-Min-Plus Product]
Let \emph{$d$-column weight $(n,D)$-Min-Plus Product} be the following problem: Given an $n\times D$ matrix $A$ and an $D \times n$ matrix $B$ such that every {\em column} of $A$ and every row of $B$ have at most $d$ distinct integer entries, compute the Min-Plus product of $A$ and $B$.
\end{definition}

\begin{definition}[$d$-Row Weight $(n,D)$-Min-Plus Product]
Let \emph{$d$-row weight $(n,D)$-Min-Plus Product} be the following problem: Given an $n\times D$ matrix $A$ and an $D \times n$ matrix $B$ such that every {\em row} of $A$ and every column of $B$ have at most $d$ distinct integer entries, compute the Min-Plus product of $A$ and $B$.
\end{definition}

Notice that both $n^{1-\mu}$-column weight $(n,n^\mu)$-Min-Plus Product and $n^{1-\mu}$-row weight $(n,n^\mu)$-Min-Plus Product are strict generalizations of $(n, n^\mu, n \mid n^{1-\mu})$-Min-Plus Product which Chan, Vassilevska W. and Xu \cite{ChanWX21} proved is equivalent to u-dir-APSP. In particular, both $n^{1/2}$-column weight $(n,n^{1/2})$-Min-Plus and $n^{1/2}$-row weight $(n,n^{1/2})$-Min-Plus require $n^{2.5-o(1)}$ time under the u-dir-APSP Hypothesis.

We first show that node-weighted APSP is at least as hard as 
$n^{1-x}$-column weight $(n,n^x)$-Min-Plus for every $x$.

\begin{lemma}[Reduction from Column Weight Min-Plus Product]
\label{lem:merge2}
For any $x\in [0,1]$, $n^{1-x}$-column weight $(n,n^x)$-Min-Plus Product can be reduced in $O(n^{2})$ time to a single instance of node-weighted APSP in an $O(n)$ node directed or undirected graph. (In particular, under the u-dir-APSP hypothesis, undirected node-weighted APSP requires $n^{2.5-o(1)}$ time.)
\end{lemma}

\begin{proof}
Let $A$ and $B$ be the $n\times n^x$ and $n^x\times n$ matrices given, respectively with at most $n^{1-x}$ distinct entries per $A$-column and $B$-row.

For every $k\in [n^x]$, let $A(k)$ be the set of distinct entries in column $k$ of $A$ and let $B(k)$ be the set of distinct entries in row $k$ of $B$. By assumption, $|A(k)|,|B(k)|\leq n^{1-x}$.

We create a graph $G$ on $4$ layers of nodes: $I,K_1,K_2,J$.
We first describe $G$ as a directed graph. (See Figure~\ref{fig:colred}.)

\begin{figure}[t]
\centering
\includegraphics[width=8cm]{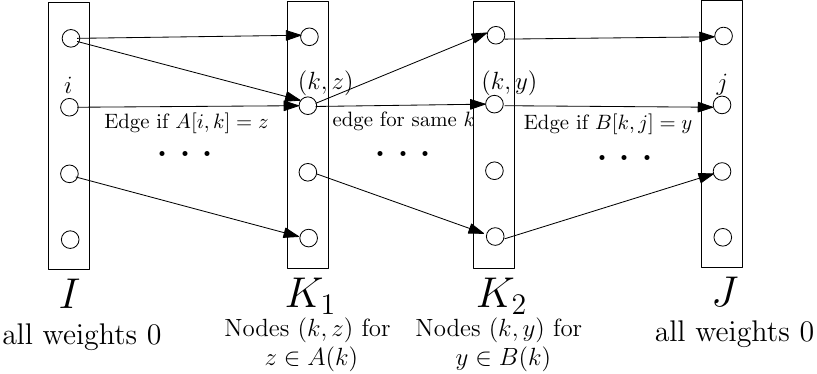}
\caption{Our reduction from $d$-column weight $(n,n/d)$-Min-Plus Product to node-weighted APSP.}
\label{fig:colred}
\end{figure}

\begin{itemize}
\item $I$ has a node $i$ for every row $i$ of $A$, with weight $0$.
\item $J$ has a node $j$ for every column $j$ of $B$, with weight $0$.
\item $K_1$ has a node $(k,z)$ with weight $z$ for every pair $k\in [n^x]$, $z\in A(k)$. 
\item $K_2$ has a node $(k,y)$ with weight $y$ for every pair $k\in [n^x]$, $y\in B(k)$.
\end{itemize}

Every layer has $\leq n$ nodes by construction.

The edges are as follows:
\begin{itemize}
\item For every $k\in [n^x]$ and every $z\in A(k), y\in B(k)$, add an edge from $(k,z) \in K_1$ to $(k,y) \in K_2$
\item For every $i\in [n],k\in [n^x]$, add an edge from $i\in I$ to $(k,A[i,k]) \in K_1$.
\item For every $j\in [n],k\in [n^x]$,  add an edge from  $(k,B[k,j])\in K_2$ to $j \in J$.
\end{itemize}

The construction time is $O(n^{2})$.
The distance from $i\in I$ to $j\in J$ in $G$ is $\min_k\{ A[i,k]+B[k,j]\}$, and thus if directed Node-weighted APSP can be solved in $O(n^c)$ time, so can $n^{1-x}$-column weight $(n,n^x)$-Min-Plus.

To make the graph undirected, additionally, add a weight of $3M$ to every node where $M$ is the maximum finite entry in $A$ or $B$. Assuming that all entries are nonnegative, this will make any undirected path on at least $5$ nodes have weight at least $15M$, whereas any path on $4$ nodes has weight at most $14M$. Hence the distance from $i\in I$ to $j\in J$ in $G$ is $12M+\min_k \{A[i,k]+B[k,j]\}$, and thus if undirected Node-weighted APSP can be solved in $O(n^c)$ time, so can $n^{1-x}$-column weight $(n,n^x)$-Min-Plus.
\end{proof}

Our reduction provides an $\widetilde{O}(n^{(3+\omega)/2})$ time algorithm for $n^{1-x}$-column weight $(n,n^x)$-Min-Plus for every $x\in [0,1]$. 

Let's now consider $n^{1-x}$-row weight $(n,n^x)$-Min-Plus.
When $x\leq 1/2$, we have that $n^{1-x}\geq n^x$ so that $n^{1-x}$-row weight $(n,n^x)$-Min-Plus is just equivalent to Min-Plus product of an arbitrary $n\times n^x$ integer matrix by an arbitrary $n^x \times n$ integer matrix. We obtain:

\begin{observation}
For any $x \leq 1/2$, the $n^{1-x}$-row weight $(n, n^x)$-Min-Plus Product problem cannot be solved in time $O(n^{2+x-\Omega(1)})$, unless the APSP Hypothesis fails.
\end{observation}

Thus already for $x\leq 1/2$, the row-weight version is much harder than the corresponding column-weight version. 

In \cref{sec:apsp-few-weights} we showed that $n\times n$ Min-Plus product can be solved in truly subcubic time whenever the number of distinct entries per row in matrix $A$ is $n^{1-\eps}$ for some $\eps>0$. Here we show that there is a relationship between the rectangular version of the small number of distinct entries per row of $A$ Min-Plus and node-weighted APSP.

We show that an $O(n^{2+x-\eps})$ time algorithm for node-weighted APSP for any $x>1/2$ and $\eps>0$ would imply that (for that $x$) $n^{1-x}$-row-weight $(n,n^x)$-Min-Plus can be solved in $O(n^{2+x-\eps'})$ time for $\eps'>0$.

In particular, via the reduction we get that $n^{1-x}$-row-weight $(n,n^x)$-Min-Plus can be solved in $O(n^{2+x-\eps'})$ time for $\eps'>0$ whenever $x> (\omega-1)/2$. 

If $\omega=2$, this would imply that for every $x>1/2$, $n^{1-x}$-row-weight $(n,n^x)$-Min-Plus can be solved in $O(n^{2+x-\eps'})$ time, whereas for $x\leq 1/2$ getting such an algorithm is equivalent to obtaining truly subcubic APSP.

\begin{theorem}[Reduction from Row Weight Min-Plus Product]
If there is some $x\in (1/2,1)$ such that Node-weighted APSP in directed or undirected graphs is in $O(n^{2+x-\eps})$ time for some $\eps>0$, we get that $n^{1-x}$-row-weight $(n,n^x)$-Min-Plus can be solved in $O(n^{2+x-\eps'})$ time for $\eps'>0$.
\end{theorem}

\begin{proof}
Let $d = n^{1-x}$.
Let $A$ and $B$ be the $n\times n/d$ and $n/d\times n$ matrices given, respectively. 

First, using a standard scaling trick\footnote{Basically, we recursively compute the Min-Plus product of $A',B'$ where $A'[i,k]=\lfloor A[i,k]/2\rfloor, B'[k,j]=\lfloor B[k,j]/2\rfloor$ and then set $C[i,j]=2\min_k \{ A'[i,k]+B'[k,j]\}$. The scaling only decreases the number of distinct entries per row/column.}, we can assume that in addition to $A$ and $B$, we are given an $n\times n$ matrix $C$, with the promise that for every $i,j$, $\min_k \{A[i,k]+B[k,j]\}\in \{C[i,j],C[i,j]+1,C[i,j]+2\}$.

Second, we will solve $O(\log^2 n)$ instances of the following problem: $A$ and $B$ are as above, but every row in $A$  has $\leq d_A\leq d$ distinct finite entries and for every finite entry $w$, the number of entries $A[i,k]$ of row $i$ of $A$ equal to $w$ is at most $O((n/d)/d_A)$. Similarly, every column of $B$ has $\leq d_B\leq d$ distinct finite entries and for every finite entry $w$, the number of entries $B[k,j]$ of column $j$ of $B$ equal to $w$ is at most $O((n/d)/d_B)$.
(This is done similarly to our earlier sections, e.g., Step 1 in the proof of \cref{lem:exact-tri-uniformization}.)  By symmetry, assume $d_A\le d_B$ without loss of generality.

Let $\Delta\geq 1$ be a parameter to be determined later.  If $d_B > d_A\sqrt{\Delta}$, then we run the following brute force algorithm: for every $i,j\in [n]$ and candidate answer $c\in \{C[i,j],C[i,j]+1,C[i,j]+2\}$, enumerate the $\le d_A$ many possibilities $x$ of the value of $A[i,k]$, and enumerate all the $\le (n/d)/d_B$ many indices $k$ with $B[k,j]=c-x$, and then check whether $A[i,k]=x$ holds. The total running time is $\widetilde O(n^2\cdot d_A \cdot (n/d)/d_B) = \widetilde O(n^3/(d\sqrt{\Delta}))$. Similarly, for the case of $d_A > d_B\sqrt{\Delta}$, we can also run the same brute-force algorithm in $\widetilde O(n^3/(d\sqrt{\Delta}))$ time. Now it remains to consider the case where $d_A\le d_B\le d_A\sqrt{\Delta}$.

For every $i\in [n]$ define $S_i$ to be the set of weights present in the $i$th row of $A$ and let $T_i$ be the set of weights in the $i$th column of $B$.
Note that $|S_i|,|T_i|\leq \max\{d_A,d_B\} = d_B\le d$ for all $i$.

Apply Lemma~\ref{lem:pop-sums-decomp} to these sets with parameter $\Delta$ to obtain, in $\widetilde{O}(n^2 d_B\Delta^3)=\widetilde{O}(n^2 d\Delta^3)$ time,
partitions
\begin{alignat*}{2}
    S_i &= S_{i, 1} \sqcup \dots \sqcup S_{i, \Delta^2} \sqcup S'_i \qquad &&\text{(for all $i \in [n]$),} \\
    T_j &= T_{j, 1} \sqcup \dots \sqcup T_{j, \Delta^2} \sqcup T'_j \qquad &&\text{(for all $j \in [n]$)}
\end{alignat*}
and for all $\ell\in [\Delta^2]$, sets $X_\ell, Y_\ell$ and shifts $\sigma_{i,\ell}, \tau_{j,\ell}$ so that 
\begin{center}$S_{i,\ell}\subseteq \sigma_{i,\ell}+X_\ell$ and $T_{j,\ell}\subseteq \tau_{j,\ell}+Y_\ell$, and\end{center}
so that there are at most $n^2/\Delta$ pairs $(i,j)$ such that
for some integer $c$ there are at least $2d_B/\Delta$ values $x\in S'_{i}, c-x\in T_j$ or $x\in S_{i}, c-x\in T'_j$.

We begin by dealing with the ``unpopular'' case: For every $i,j\in [n]$ and $c\in \{C[i,j],C[i,j]+1,C[i,j]+2\}$, in $O(n^2 d)$ time search through $S'_i$ and $T'_j$ to find $\leq O(d_B/\Delta)$ pairs $(x,y)\in (S'_i \times T_j)\cup (S_i \times T'_j)$ with $x+y=c$. Then, since there are at most $O(n/(dd_A))$ 
columns $k$ such that $A[i,k]=x$ for each $x$ in one of the pairs found, we can deal with this case in time 
$$O(n^2d+ n^2\cdot (d_B/\Delta)\cdot n/(dd_A)) =  O(n^2d + n^3/(d\sqrt{\Delta})),$$  
where we used $d_B\le d_A\sqrt{\Delta}$.
For the remaining $n^2/\Delta$ pairs $(i,j)$, we deal with them by brute force enumeration over $k\in [n/d]$, in $\widetilde O(n^3/(\Delta d))$ total time.

Now, we deal with the ``popular'' case where we consider $\min_k \{ A[i,k]+B[k,j]\}$ where the min is over $A[i,k]\in \sigma_{i,\ell} + X_\ell$ and $B[k,j]\in \tau_{j,\ell'}+Y_{\ell'}$ for some choice of $\ell,\ell'\in  [\Delta^2]$.

For this we create $\Delta^4$ instances of Node-weighted APSP on $n$ node graphs as follows.
For every choice of $\ell,\ell'\in [\Delta^2]$, create a graph $G_{\ell,\ell'}$ on $4$ layers of nodes: $I,K_1,K_2,J$.
\begin{itemize}
\item $I$ has a node $i$ for every row $i$ of $A$, with weight $\sigma_{i,\ell}$.
\item $J$ has a node $j$ for every column $j$ of $B$, with weight $\tau_{j,\ell'}$.
\item $K_1$ has a node $(k,x)$ with weight $x$ for every pair $k\in [n/d]$, $x\in X_\ell$.
\item $K_2$ has a node $(k,y)$ with weight $y$ for every pair $k\in [n/d]$, $y\in Y_{\ell'}$.
\end{itemize}

The edges are as follows:
\begin{itemize}
\item For every $k$ and every $x\in X_\ell, y\in Y_{\ell'}$, add an edge from $(k,x)$ to $(k,y)$
\item For every $i\in [n],k\in [n/d]$, if $A[i,k]\in \sigma_{i,\ell} + X_\ell$, add an edge from $i$ to $(k,A[i,k]-\sigma_{i,\ell})$.
\item For every $j\in [n],k\in [n/d]$, if $B[k,j]\in \tau_{j,\ell'}+Y_{\ell'}$, add an edge from  $(k,B[k,j]-\tau_{j,\ell'})$ to $j$.
\item All edges above are directed. To make the graph undirected, additionally, add a weight of $10M$ to every node where $M$ is the maximum absolute value of any finite entry in $A$ or $B$. 
This will make any undirected path on at least $5$ nodes have weight at least $45M$, whereas any path on $4$ nodes has weight at most $44M$.
\end{itemize}
The distance from $i\in I$ to $j\in J$ in $G_{\ell,\ell'}$ is \[\min \{A[i,k]+B[k,j]~|~k\in [n/d], A[i,k]\in \sigma_{i,\ell} + X_\ell, B[k,j]\in \tau_{j,\ell'}+Y_{\ell'}\}.\]
The number of nodes in this node-weighted graph is $O(n)$. Hence by solving node-weighted APSP in $\widetilde{O}(n^c \Delta^4)$ time in all $G_{\ell,\ell'}$ we solve the ``popular'' case.

The total running time is:
\[
    \tilde O\parens*{n^2d\Delta^3 + n^3/(d\sqrt{\Delta})+n^c \Delta^4}.
\]
Recall that we set $d=n^{1-x}$ where $x\in (1/2,1)$. The running time becomes:
\[
    \tilde O\parens*{n^{3-x}\Delta^3 + n^{2+x}/\sqrt{\Delta}+n^c \Delta^4}.
\]
We are interested in when this running time is $O(n^{2+x-\eps'})$ for some $\eps'>0$. If we set $\Delta=n^\delta$ for $\delta>0$, the second term in the running time is $O(n^{2+x-\delta/2})$, so we can disregard it. The remaining running time is:
\[
    \tilde O\parens*{n^{3-x+3\delta} + n^{c+4\delta}}.
\]
We want $3-x+3\delta\leq 2+x-\delta'$, so we should set $3\delta\leq 2x-1-\delta'$, i.e. $\delta=(2x-1-\delta')/3$.
This is possible when $x>1/2$. That is, when $x=1/2+q$ for some $1/2> q>0$, we can set $\delta'=2q-3\delta$ for arbitrarily small $\delta>0$. Then, $\delta'\in (0,1)$.  The running time becomes $O(n^{2+x-\eps})$ for $\eps=\min\{\delta',\delta\}$, plus an extra $\widetilde{O}(n^{c+4\delta})$ time. We would like that
\[c+4\delta< 2+x.\]
Thus, if $c<2+x$, we can pick $\delta>0$ small enough so that the running time is $O(n^{2+x-\eps})$ for some $\eps>0$. Hence if there is some $x\in (1/2,1)$ such that node-weighted APSP is in $O(n^{2+x-\eps})$ time for some $\eps>0$, we get that $n^{1-x}$-row-weight Min-Plus can be solved in $O(n^{2+x-\eps'})$ time for $\eps'>0$.
\end{proof}

\bibliographystyle{alphaurl}
\bibliography{paper}

\end{document}